\documentclass[a4paper,UKenglish,cleveref, autoref, thm-restate]{lipics-v2021}

\makeatletter
\def\@hideLIPIcs{1} %
\makeatother

\usepackage{hyperref}
\usepackage{tabularx}

\usepackage{amsfonts} 
\usepackage{mathtools} 
\usepackage{graphicx} 
\usepackage{algorithmic} 
\usepackage{algorithm} 
\usepackage{ulem} 
\usepackage{multicol}
\usepackage[referable]{threeparttablex} 
\usepackage{tabu} 
\usepackage{multirow} 
\usepackage{makecell} 
\usepackage{longtable} 
\usepackage{booktabs}
\usepackage{pifont} 
\usepackage{wasysym}
\usepackage{tikz} 
\usetikzlibrary{arrows.meta,positioning, tikzmark,calc}
\usepackage{subcaption} 
\usepackage{wrapfig}

\usepackage{mdframed}

\usepackage{xcolor} 
\usepackage{colortbl} 
\usepackage{caption} 
\usepackage{calc} 
\usepackage{pdflscape} 
\usepackage{supertabular} 
\usepackage{numprint} 
\npdecimalsign{.}
\usepackage{placeins}

\definecolor{green}{RGB}{0, 150, 130}
\definecolor{green70}{RGB}{76, 181, 167}
\definecolor{blue}{RGB}{70, 100, 170}
\definecolor{blue70}{RGB}{125,146,195}
\definecolor{blue50}{RGB}{162,177,212}                                                           
\definecolor{blue30}{RGB}{199,208,229}
\definecolor{blue15}{RGB}{227,232,242}
\definecolor{lightgray}{rgb}{0.86,0.86,0.86}
\definecolor{greengray}{RGB}{193,228,224}
\definecolor{redgray}{RGB}{233,183,183}
\definecolor{red}{RGB}{204,0,0}

\newcommand{\Lin}{L_{in}}
\newcommand{\Lout}{L_{out}}%
\newcommand{\ILP}{\textsc{eILP}}
\newcommand{\mstconnect}{\textsc{MSTConnect}}%
\newcommand{\gwc}{\textsc{GWC}}%
\newcommand{\ls}{\textsc{LS($k$)}}%
\newcommand{\lsthree}{\textsc{LS(3)}}%
\newcommand{\lsfive}{\textsc{LS(5)}}%

\newcommand{\FSM}{\textsc{FSM}}%
\newcommand{\HBD}{\textsc{HBD}}%
\newcommand{\SMC}{\textsc{SMC}}%

\newcommand{\ie}{i.e.\ }
\newcommand{\etal}{et~al.}

\nolinenumbers

\usepackage[]{todonotes}

\title{Engineering Weighted Connectivity Augmentation Algorithms
	 }

\author{Marcelo Fonseca Faraj}{Heidelberg University, Heidelberg, Germany}{marcelofaraj@informatik.uni-heidelberg.de}{https://orcid.org/0000-0001-7100-236X}{}
\author{Ernestine Großmann}{Heidelberg University, Heidelberg, Germany}{e.grossmann@informatik.uni-heidelberg.de}{https://orcid.org/0000-0002-9678-0253}{}
\author{Felix Joos}{Heidelberg University, Heidelberg, Germany}{joos@informatik.uni-heidelberg.de}{
	https://orcid.org/0000-0002-8539-9641}{}
\author{Thomas Möller}{Heidelberg University, Heidelberg, Germany}{thomas.moeller@uni-heidelberg.de}{https://orcid.org/0009-0005-4312-2272}{}
\author{Christian Schulz}{Heidelberg University, Heidelberg, Germany}{christian.schulz@informatik.uni-heidelberg.de}{https://orcid.org/0000-0002-2823-3506}{}

\authorrunning{M. Fonseca \etal} %

\Copyright{Marcelo Fonseca Faraj, Ernestine Großmann, Felix Joos, Thomas Möller, Christian Schulz} %


\category{} %

\relatedversion{} %

\acknowledgements{We acknowledge support by DFG grant SCHU 2567/3-1.}

\EventEditors{}
\EventNoEds{0}
\EventLongTitle{}
\EventShortTitle{}
\EventAcronym{}
\EventYear{}
\EventDate{}
\EventLocation{}
\EventLogo{}
\SeriesVolume{}
\ArticleNo{}

\begin{document}
	
	\maketitle

\begin{abstract}
Increasing the connectivity of a graph is a pivotal challenge in robust network design. The weighted connectivity augmentation problem is a common version of the problem that takes link costs into consideration. The problem is then to find a minimum cost subset of a given set of weighted links that increases the connectivity of a graph by one when the links are added to the edge set of the input instance.
In this work, we give a first implementation of recently discovered better-than-2 approximations. Furthermore, we propose three new heuristic and one exact approach. These include a greedy algorithm considering link costs and the number of unique cuts covered, an approach based on minimum spanning trees and a local search algorithm that may improve a given solution by swapping links of paths. Our exact approach uses an ILP formulation with efficient cut enumeration as well as a fast initialization routine. We then perform an extensive experimental evaluation which shows that our algorithms are faster and yield the best solutions compared to the current state-of-the-art as well as the recently discovered better-than-2 approximation algorithms.
Our novel local search algorithm can improve solution quality even further. 
\end{abstract}

\maketitle

\vfill

\newpage

\setcounter{page}{1}%
\section{Introduction}\label{sec:intro}
Many real-world coherences can be modeled as graphs, including technological, social, and biological networks. A common problem of interest is the robustness of such a graph. Particularly in technological networks this is important for creating systems that are robust and fail-safe~\cite{robustness}. An example is a power grid where single lines can fail, either randomly due to age, or by targeted attacks. If a line fails, alternative routes are used which is increasing the load on them and therefore the chance of failure. To obtain a fail-safe network that can survive both, random failures and targeted failures of important lines, the graph needs to be well-connected. Increasing the connectivity and therefore improving the robustness at minimum cost is known as connectivity augmentation or the survivable network problem.
Another technological example is a computer network like the internet which should be designed in a fail-safe way while reliable transportation networks can avoid \hbox{traffic congestion}.

However, it is well-known that the weighted connectivity augmentation problem is NP-hard. Eswaran and Tarjan~\cite{np-complete} have shown that the decision problem, whether there is an augmentation of at most a given weight, is NP-complete. Frederickson and Ja'Ja'~\cite{heuristic-arborescence} have shown that this is also true for the simpler special case where the graph is a tree, with weights being only 1 or 2. This justifies the importance of good heuristic and approximation algorithms. Furthermore, the weighted connectivity augmentation problem is APX-hard, which was also shown for the weighted tree augmentation problem by Kortsarz, Krauthgamer and Lee~\cite{apx-hardness}. Despite the fact that there is no polynomial time approximation algorithm with an approximation factor arbitrarily close to 1, there has been much progress in improving the \hbox{approximation ratio}.
Recently, the connectivity augmentation problem has been discussed frequently in the context of approximation algorithms with approximation factors \hbox{below~2~\cite{unweighted-below-2,unweighted-below-2-2,wtap,wcap}.} This includes work on special cases like the tree augmentation problem~\cite{wtap}, \hbox{as well as the general case~\cite{unweighted-below-2,wcap}.}

In recent years, Henzinger~\etal~\cite{mincut,mincutjv} developed the leading codes for the minimum cut problem in graphs. This includes the development of cutting-edge shared-memory inexact algorithms, consistently delivering near-optimal results. Additionally, they engineered state-of-the-art shared-memory exact algorithms~\cite{viecut}, surpassing the previous state-of-the-art by an order of magnitude in running time, as well as highly efficient approaches for tackling the broader all minimum cut problem~\cite{cactus-state-of-the-art}. It turns out that these algorithms, i.e.~computing minimum cuts, (enumeration of) all minimum cuts, and the efficient computation of a cactus representation of a graph, are important subroutines for algorithms that tackle the connectivity augmentation problem. Thus in this work, we heavily employ these recently developed techniques to engineer efficient algorithms for the \hbox{connectivity augmentation problem.}

 \textit{Our Results.} Our contribution in this work is two-fold. 
 First, we give the first implementation and experimental evaluation of two recently discovered connectivity augmentation approximation algorithms due to Traub and Zenklusen~\cite{wcap}. More precisely, Traub and Zenklusen describe two algorithms: a greedy $(1+\ln2+\epsilon)$-approximation and the local search based $(1.5+\epsilon)$-approximation, which we implement and evaluate.
 	
Secondly, we propose three new heuristic and an exact algorithm. The first algorithm is a greedy heuristic considering link costs and the number of cuts covered by a link. This simple algorithm already outperforms all previous state-of-the-art algorithms by more than \numprint{22}\,\% %
 improvement in solution quality (reduced link costs) on instances where links have small costs. Our second strategy uses minimum spanning trees to find a feasible solution first and then greedily improves it. Additionally, we present a local search algorithm that can improve a given solution by replacing link sets with cheaper ones. On instances with large link costs, the minimum spanning tree algorithm has the overall best performance regarding solution quality, running time and memory consumption. 
 It computes solutions \numprint{8}\,\% better than the best performing previous state-of-the-art on these instances, while being a factor of \numprint{7} times faster on average. With our local search algorithm we can further improve these solutions on average by \numprint{2}\,\%. Lastly, we introduce a new exact solver using an ILP formulation with efficient cut enumeration as well as a fast initialization routine, for which we utilize our fast minimum spanning tree heuristic. Especially on real-world instances, it is able to outperform the previous state-of-the-art heuristic solvers regarding running time for both \hbox{small and large link costs.}

 \section{Preliminaries}\label{sec:prelim}

 An undirected graph $G=(V,E)$ is a structure that consists of a set of vertices $V$ and a set of edges $E\subseteq \binom{V}{2}$ connecting pairs of vertices. The number of vertices is denoted as $n$ and the number of edges as $m$. The graph $G$ is \emph{connected} if there is a path between any two vertices.
 The \emph{edge connectivity} of a graph is the maximal number of edge-disjoint paths that exist between any pair of vertices. A graph is \emph{$k$-edge-connected}, if ${k-1}$ arbitrary edges can be removed without disconnecting the graph. %
 A partition of a graph is a partition of the vertex set into mutually disjoint sets.
 
 A \emph{cut} of a graph is a partition of the vertex set into two disjoint subsets, also called a bipartition. Any cut can be represented as one of its two constituent vertex sets. Every non-empty proper subset of $V$ is a cut.  To prevent different representations of the same cut we use the notation where a cut is given as one set of the partition (only the representation that does not include an arbitrarily chosen root $r\in V(G)$ is used).
 The \emph{size} or \emph{weight} of a cut is the number of edges or the sum of the edge weights that have one endpoint in each subset. A cut is a \emph{minimum cut} if there is no cut with smaller size or weight. The set of all minimum cuts is denoted as $C_G$ and $\operatorname{cut}:C_G\times \binom{V}{2} \to\{0,1\}$ is a function that is 1 if and only if the endpoints $u$ and $v$ of an edge $e=uv\in\binom{V}{2}$ lie in different sets of the partition of a cut $c\in C_G$.
 
The goal of the \textit{Weighted Connectivity Augmentation Problem (WCAP)} is to increase the edge connectivity of a graph. More formally, for a given $k$-connected graph $G$ with a set of links $L\subseteq \binom{V}{2}$ and a cost function $c:L\to\mathbb{R}_{\geq0}$, the task is to find the cheapest subset of links $S\subseteq L$ that will increase the edge connectivity to $k+1$. A link $l\in L$ \emph{covers} a minimum cut $c\in C_G$ if the size or weight of the cut $c$ is increased in the graph $G'=(V,E\cup\{l\})$. The graph $G_L=(V,L)$ is called the link graph. The set of links is disjoint with the set of edges, \ie $L\cap E = \emptyset$.
 For the ease of notation the cost function is extended to sets, where it is the sum of the cost of all elements.
 
 If the input graph is disconnected, the weighted connectivity augmentation problem coincides with the minimum spanning tree problem among its components. In this work we therefore only focus on connected graphs, as the other case is simple to solve via well-known efficient minimum spanning tree algorithms. 
 
 A \textit{cactus graph} is a connected graph, such that any two cycles have at most one vertex in common. To distinguish between edges that lie within a cycle and those that do not, they are called cycle edges and tree edges, respectively.
 The \textit{cactus graph representation} of the set of minimum cuts $C_G$ of a graph $G=(V,E)$ is a cactus graph $C= (V_c,E_c)$ with a function $\Pi:V\to V_c$ and its inverse $\Pi^{-1}:V_c\to2^{V}$ which is defined as $v \mapsto\{u\in V:\Pi(u)=v\}$. The functions $\Pi$ and $\Pi^{-1}$ are defined such that each minimum cut in $C$ corresponds to a minimum cut in $G$, \ie for all $c\in C_C:\bigcup_{v\in c}\Pi^{-1}(v)\in C_G$ and each minimum cut in $G$ is represented in the cactus graph, \ie for all $c_G\in C_G$ there exists a $c_C\in C_C$ such that $\Pi(v)\in c_C$.
  Figure~\ref{fig:cactus-contraction} gives an example for a graph, with its minimum cuts and the corresponding cactus graph representation. 
  Dinitz~et~al.~\cite{Dinitz2011OnTS} have shown that all minimum cuts of a connected graph $G=(V,E)$ can be represented as a cactus graph $C=(V_c,E_c)$. For more details on the computation of the cactus graph, we refer the reader to~\cite{ cactus-state-of-the-art}. %
   Analogous to the link graph $G_L$, we define the cactus link graph $C_L=(V_c,L_c)$, with $L_c\subset L$ where the links are mapped from $G$ to $C$ and for each vertex pair we only keep the cheapest link.
 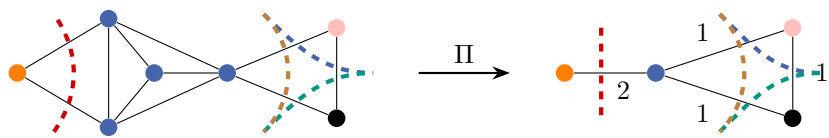
\begin{figure}[bt!]
	\caption{A graph and its weighted cactus graph with corresponding minimum cuts drawn as dashed lines of same color. Vertex colors encode the function $\Pi$.}
	\label{fig:cactus-contraction}
    \centering
    \begin{tikzpicture}[scale=1.2, main/.style = {fill,circle,scale=0.7}]
        \node[main,orange] (1) at (0, 0) {};
        \node[main,blue] (2) at (1, -0.6) {};
        \node[main,blue] (3) at (1, 0.6) {};
        \node[main,blue] (4) at (2.3, 0.) {};
        \node[main,blue] (5) at (1.5, 0.) {};
        \node[main,pink] (6) at (3.5, 0.5) {};
        \node[main] (7) at (3.5, -0.5) {};
        \coordinate (1-2) at (0.4, -0.65);
        \coordinate (1-3) at (0.4, 0.65);
        \coordinate (4-6) at (2.7, 0.65);
        \coordinate (4-7) at (2.7, -0.65);
        \coordinate (6-7) at (3.9, 0);
        
        \draw (1) -- (2);
        \draw (1) -- (3);
        \draw (2) -- (3);
        \draw (2) -- (4);
        \draw (3) -- (4);
        \draw (2) -- (5);
        \draw (3) -- (5);
        \draw (4) -- (5);
        \draw (4) -- (6);
        \draw (4) -- (7);
        \draw (6) -- (7);
        \draw[dashed,color=red,ultra thick] (1-2) to[out=55, in=-55](1-3);
        \draw[dashed,color=blue,ultra thick] (4-6) to[out=-45, in=180](6-7);
        \draw[dashed,color=green,ultra thick] (4-7) to[out=45, in=180](6-7);
        \draw[dashed,color=brown,ultra thick] (4-6) to[out=315, in=45](4-7);

        \coordinate (as) at (4.4,0);
        \coordinate (ae) at (5.4,0);
        \draw[-Stealth,thick, scale=1.2] (as) -- (ae) node[midway,above] {\text{$\Pi$}};

        \node[main,orange] (8) at (6, 0) {};
        \node[main,blue] (9) at (7, 0.) {};
        \node[main,pink] (10) at (8.5, 0.5) {};
        \node[main] (11) at (8.5, -0.5) {};
        \coordinate (8-9-top) at (6.4, 0.5);
        \coordinate (8-9-bottom) at (6.4, -0.5);
        \coordinate (9-10) at (7.7, 0.65);
        \coordinate (9-11) at (7.7, -0.65);
        \coordinate (10-11) at (8.8, 0);

        \draw (8) -- node [below,xshift=5]{2} (9);
        \draw (9) -- node [above,xshift=-8]{1} (10);
        \draw (9) -- node [below,xshift=-8]{1} (11);
        \draw (10) -- node [right,xshift=5]{1} (11);
        \draw[dashed,color=red,ultra thick] (8-9-top) to (8-9-bottom);
        \draw[dashed,color=blue,ultra thick] (9-10) to[out=-45, in=180](10-11);
        \draw[dashed,color=green,ultra thick] (9-11) to[out=45, in=180](10-11);
        \draw[dashed,color=brown,ultra thick] (9-10) to[out=315, in=45](9-11);
    \end{tikzpicture}
\end{figure}

 \section{Related Work} \label{sec:related_work}
This paper is a summary and extension of the master theses
\cite{wca-thesis}. In this section, we present the state-of-the-art for computing all minimum cuts of a graph in the cactus graph representation, followed by research and implementations in the field of connectivity augmentation problems. 

\subsection{Minimum Cuts}
Computing all minimum cuts is usually a fundamental step in connectivity augmentation.
Nagamochi, Nakao and Ibaraki presented an efficient algorithm to compute all minimum cuts in the cactus graph representation~\cite{Nagamochi2000AFA}. A cactus representation can be computed in ${\mathcal{O}(mn + n^2 \log n + n^*m \log n)}$ time where $n^*$ is the number of cycles in the cactus representation. They observed that all minimum cuts between two vertices $s$ and $t$ can be computed by running a maximum $s$-$t$-flow algorithm, and edges that are cut by no minimum cut \hbox{can be contracted.}

The current state-of-the-art algorithm to compute all minimum cuts is \textsc{VieCut} by Henzinger~\etal~\cite{viecut,cactus-state-of-the-art}. It uses linear time edge contraction based reduction rules and an optimized version of the algorithm by Nagamochi, Nakao and Ibaraki. For example, an edge $uv$ can be contracted if the connectivity between $u$ and $v$ is larger than the minimum cut. Such edges could be found by computing $k$ edge-disjoint spanning trees where $k$ is the size of the minimum cut~\cite{viecut,nagamochi-mincut}. Furthermore, reduction rules by Padberg and Rinaldi~\cite{mincut-reduction-rules} were adapted from the problem of finding one minimum cut to the problem of finding all minimum cuts. Lastly, edges that form a trivial minimum cut are contracted and remembered. These cuts are reintroduced at the end of the algorithm. The reduction rules are used exhaustively as long as a significant number of edges is contracted. The remaining kernel is solved based on the algorithm by Nagamochi, Nakao and Ibaraki~\cite{Nagamochi2000AFA}. In this work, we use \textsc{VieCut}~\cite{viecut,cactus-state-of-the-art} to compute the cactus representation of a graph, which has a much better performance in practice than what the worst-case analysis predicts. Henceforth, when we analyse the complexity of an algorithm, we do not include the complexity to compute the cactus, as this is the same for all algorithms.\vfill

\subsection{Connectivity Augmentation}
There have been several approximation algorithms for the connectivity augmentation problem in the past. An early approach is using minimum cost arborescences, which was introduced by Frederickson and Ja'Ja'~\cite{heuristic-arborescence} for bridge connectivity augmentation, the case where the graph is 1-connected but not 2-connected, and generalized for the WCAP by Watanabe~\etal~\cite{watanabe-heuristics}. The bridge connectivity algorithm results in a 2-approximation while the generalization \hbox{cannot guarantee an approximation factor.}
There are well-known 2-approximations for the WCAP. One possibility was discovered in 1992 and reduces the problem to a directed version by replacing each undirected edge with two directed edges~\cite{2-approx_1}. The directed version can be solved in polynomial time based on minimum-cost flows~\cite{2-approx-flow} or by using a linear program which has integral solutions for the cactus augmentation problem~\cite{tap}.
Another approach involves the LP relaxation of an ILP formulation combined with iterative \hbox{rounding techniques~\cite{2-approx-ilp}}. %

Only recently, progress has been made on various approximation algorithms regarding special cases of the connectivity augmentation problem, as well as the general case. \hbox{For the unweighted} version of the connectivity augmentation problem the first approximation with factor below 2 was found in 2020 by Byrka~\etal~\cite{unweighted-below-2,unweighted-below-2-2}. The WCAP is reduced to the steiner tree problem, for which a specialized approximation gives an \hbox{approximation factor of 1.91.}

For the tree augmentation problem (TAP), where the cactus graph is a tree, and the unweighted connectivity augmentation problem an approximation factor of 1.393 was found in 2021~\cite{tap}. For the weighted tree augmentation problem, Traub and Zenklusen~\cite{wtap} discovered a (${1+\ln 2+\epsilon}$)-approximation, which builds upon the 2-approximation reducing the problem to a directed one and greedily improves this solution. Afterwards, they transferred the algorithm to the weighted connectivity augmentation problem and refined it to a ($1.5+\epsilon$)-approximation~\cite{wcap}, which improves an arbitrary solution through local search.
However, no implementations or experimental results of those algorithms existed until this point.

There has been recent work on randomized Monte Carlo algorithms that give a solution with high probability on graphs with integer edge weights based on maximum flow computations.
The state-of-the-art is an $\tilde{\mathcal{O}}(m)$ time algorithm that gives a near-linear running time by Cen~\etal~\cite{near-linear}. This shows that the connectivity augmentation problem is simpler than the maximum flow problem as there is no known $\tilde{\mathcal{O}}(m)$ time maximum flow algorithm.

\textbf{Experimentally Evaluated Algorithms.}
There have been practically applicable heuristic algorithms in the past, however, there has not been much progress on the WCAP recently. Watanabe~\etal~\cite{watanabe-fifth-apx,watanabe-matching,watanabe-heuristic-2,watanabe-heuristics} proposed five different approaches, called \textsc{FSA}, \textsc{MW}, {\FSM}, {\SMC} and {\HBD}, including experimental evaluation.
An observation used for all algorithms is that there is a subset of all vertices of the cactus graph representation that must be an endpoint in any augmentation. The algorithm \textsc{FSA} uses minimum cost arborescences based on the ideas of Frederickson and Ja'Ja'~\cite{heuristic-arborescence}. \textsc{MW} is a 2-approximation also based on arborescences. {\FSM} is based on maximum cost matchings. 
The third approach, {\SMC}, is a greedy strategy adding the cheapest incident link for each vertex of the cactus graph representation. {\HBD} is a combination of {\FSM} and {\SMC}.
Experimental results showed that the solution quality of {\FSM} is the best, followed by {\HBD}, {\SMC}, \textsc{FSA} and lastly \textsc{MW}~\cite{watanabe-fifth-apx,watanabe-heuristics}. Regarding running time, {\SMC} is the fastest algorithm, followed by \textsc{FSA}, {\HBD}, {\FSM} and lastly \textsc{MW}. {\HBD} is considered the best general algorithm, because it prevents arbitrary bad solutions that may be produced by {\FSM} or {\SMC}. \textsc{MW} is the only algorithm with a guaranteed approximation factor, however, in practice it is slower and the solution quality of the other algorithms is better~\cite{watanabe-fifth-apx}. To the best of our knowledge, no other experimentally evaluated algorithms are mentioned \hbox{in the literature.} \vfill

\section{Approximation Algorithms}\label{sec:apx}
We now describe the approximation algorithms implemented in this paper in more detail. First we describe the 2-approximation~\cite{tap} followed by brief discussions on two approximation algorithms with approximation factors $(1+\ln2+\epsilon)$ and $(1.5+\epsilon)$ by Traub and Zenklusen~\cite{wcap} for which we present first implementations and experimental evaluations. For further, highly detailed information see~\cite{tap, wcap}.

\textbf{LP-based 2-Approximation.}\label{sec:apx2lp}
 For the 2-approximation each undirected link $l=uv$ is replaced by two directed links ${l_1=(u,v)}$ and ${l_2=(v,u)}$, resulting in a set $\vec{L}$ with $|\vec{L}|=2\cdot|L|$. Then, the easier directed problem can be solved. Later, a directed solution is transferred to the undirected problem by replacing each link in the directed solution with the undirected one while removing duplicates. 
 To solve the directed problem, a linear program based algorithm is implemented. This approach was first proposed by Jain~\cite{2-approx-ilp}, and Cecchetto, Traub and Zenklusen~\cite{tap} provided a definition where the solution is integral for any cactus \hbox{augmentation instance}.
 Using Brand's $\mathcal{O}(N^{2.37}\log^2N\log N/\delta)$ algorithm~\cite{lp-sota-deterministic} where $N$ is the number of variables, the time complexity of the 2-approximation is $\mathcal{O}(n^{4.74}\log^2n\log n/\delta)$.

\textbf{Relative Greedy $(1+\ln2+\epsilon)$-Approximation.}
Traub and Zenklusen~\cite{wcap} presented a greedy algorithm that improves upon the 2-approximation described above. It begins by exactly reducing a cactus graph representation to a ring graph and then replacing directed links in the solution with so called shadows, to form an arborescence. These shadows are links, that have the same weight as the original links and together still form an augmentation. The algorithm proceeds by greedily substituting sets of directed links with undirected ones resulting in mixed solutions. 
If all directed links are replaced by undirected ones, the solution is a solution to the original problem. The greedy objective is the ratio of the cost of the added undirected links and the cost of the directed links that are not needed anymore.
Since this is difficult to compute, they only consider
link sets that can be constructed iteratively by a dynamic program. Furthermore, it is easier to check if a given ratio is better or worse than the optimum. The algorithm uses binary search with respect to the ratio to determine the optimum along with a set of links that achieves this ratio. For each \hbox{bisection a dynamic program is run.}

The ${(1+\ln2+\epsilon)}$-approximation algorithm has a polynomial running time, but it is computationally expensive, especially due to the dynamic program involved. Let $\alpha=4\lceil\frac{2}{\epsilon}\rceil$, then, for graphs with $n>\alpha+2$ the size of the dynamic programming table is $\mathcal{O}(n^{2\alpha+2})$. With this, the overall computational complexity of the algorithm for integer weighted problems is $\mathcal{O}(n^{4\alpha+7}\ln(OPT))$, where $OPT$ is the optimal augmentation weight. Note that the approximation ratio can only be improved if ${1 + \ln2 +\epsilon < 2}$ requiring $\alpha \geq 28$.

\textbf{Local Search $(1.5+\epsilon)$-Approximation.}
The state-of-the-art approximation algorithm by Traub and Zenklusen~\cite{wcap} is a $(1.5+\epsilon)$-approximation. A detailed description and correctness proofs can be found in~\cite{wcap}. The algorithm is based on the ideas and the dynamic program of their relative greedy $(1+\ln2+\epsilon)$-approximation described previously.
The main difference is that the algorithm does not only greedily replace all links of a directed solution with undirected links. Instead, replaced links should iteratively improve the solution and can themselves be replaced in further iterations.
The main part of this algorithm, the dynamic program, is the same as for the $(1+\ln2+\epsilon)$-approximation. Therefore, the size of the dynamic programming table is bounded by $\mathcal{O}(n^{2\alpha+2})$. 
The total running time is bounded by $\mathcal{O}(n^{4\alpha+7}/\epsilon)$. In contrast to the $(1+\ln2+\epsilon)$-approximation this is independent of the (upscaled) augmentation weight $OPT$.

\vfill \pagebreak
\section{Weighted Connectivity Augmentation Algorithms}
This section describes our heuristic approaches and our exact algorithm for the WCAP. The data structures used for these algorithms are briefly described in Appendix~\ref{sec:appendix_data_structure}.
Unlike the approximation algorithms in the previous section, the heuristics cannot give guarantees on the solution quality, but aim at being fast or giving good solutions for many real-world~cases.
For all considered algorithms we only compute a solution on the cactus graph representation which can be transformed to the solution on the original graph, as stated in Theorem~\ref{thm:sol}. 
\begin{theorem}\label{thm:sol} Dinic, Karzanov and Lomonosov~\cite{Dinitz2011OnTS}.
	Let $G$ be a graph and $L$ the set of links. Furthermore, let $C$ be the corresponding cactus graph representation of $G$ and $C_L$ be the link graph of $C$. Then, a valid solution for the WCAP on $C$ is a valid solution to the \hbox{WCAP on $G$.}
\end{theorem}
\begin{wrapfigure}{T}{.6\textwidth}
	\vspace*{-0.85cm}
	\begin{minipage}{0.6\textwidth}
		\begin{algorithm}[H]
			\begin{algorithmic}
				\caption{{\gwc}}\label{algo:gwc}
				\STATE \textbf{input} $G=(V,E)$, $L$, $c:L\to\mathbb{R}_{\geq0}$
				\STATE \textbf{output} augmentation $S\subseteq L$
				\STATE \textbf{procedure} {\gwc}($G$, $L$, $c$)
				\STATE \quad $C=(V_c, E_c)\leftarrow $\textsc{cactusByVieCut}$(G)$
				\STATE \quad $C_L=(V_c,L_c)\leftarrow $\textsc{buildLinkGraph}$(V_c,L)$ 
				\STATE \quad $S\leftarrow\emptyset$
				\STATE \quad \textbf{while} $L_c\neq\emptyset$ \textbf{do}
				\STATE \quad\quad $l\leftarrow \operatorname{argmin}\{c(l)/$\textsc{coveredCuts}($C$, $l$) $| \ l \in L_c\}$
				\STATE \quad\quad $S\leftarrow S\cup\{l\}$
				\STATE \quad\quad $C\leftarrow$\textsc{updateCactus}($C$, $l$)
				\STATE \quad\quad $L_c\leftarrow\{l\in L_c\mid l\text{ covers a mincut in }C\}$
				\STATE \quad\textbf{return} $S$
			\end{algorithmic}
		\end{algorithm}
		\vspace*{-.8cm}
	\end{minipage}
\end{wrapfigure}

\subsection{Heuristic Algorithms}
In this section we introduce a new heuristic \textsc{GreedyWeightCoverage} ({\gwc}), where links are added to the solution greedily based on the costs per augmented cut, an algorithm based on minimum spanning trees called {\mstconnect} as well as a local \hbox{search~algorithm {\ls}}.

\textbf{{\gwc}.} As the problem aims at minimizing the cost of an augmentation, it is natural to add links of small weight. At the same time, we want to minimize the number of links added. This leads to the heuristic {\gwc}, where we greedily pick the link, for which the cost per augmented minimum cut, \ie $c(l)/a_l $ with $a_l=\textsc{coveredCuts}(C,l) \coloneqq |\{c\in C_G:\text{cut}(c, l)=1\}|$ is minimal. If $a_l=0$ the link $l$ is not considered. This regards listing all minimum cuts in the cactus graph computed using \textsc{VieCut}~\cite{viecut,cactus-state-of-the-art}.

A naive bound for the complexity of this is $\mathcal{O}(|V_c|^5)$, because the solution has at most $\mathcal{O}(|V_c|)$ links, and in each iteration the heuristic is computed for ${\mathcal{O}(|L_c|)=\mathcal{O}(|V_c|^2)}$ links by checking $\mathcal{O}(|V_c|^2)$ minimum cuts. Using our custom data structure we can run the algorithm in $\mathcal{O}(|V_c|^4)$. Due to space constraints we omit the description here, and refer to \cite{wca-thesis}. The bound is very pessimistic and in practice the performance of the algorithm is much better, especially since we use efficient  algorithms to enumerate all minimum cuts~\cite{cactus-state-of-the-art}.

We also tested other heuristics such as choosing the link with the smallest weight, such that at least one cut is covered, or picking an arbitrary uncovered minimum cut $c$ and choosing the smallest weight link that covers $c$. However, {\gwc} produced significantly better results, which is why we only report details of this algorithm here.

\textbf{{\mstconnect}.} After computing the cactus graph representation using \textsc{VieCut}~\cite{viecut,cactus-state-of-the-art}, the greedy strategies described above work by adding links to a set until this set is a valid connectivity augmentation. For {\mstconnect}, see Algorithm~\ref{algo:mst}, we have a different approach. Here, we start with a (possibly much larger) set of links that increases the connectivity when added to the cactus graph. Then, we reduce this link set, while keeping a valid solution.
\begin{wrapfigure}{}{.49\textwidth}
	\vspace*{-0.45cm}
	\hfill\begin{minipage}{.475\textwidth}
		\noindent
		\begin{algorithm}[H]
			\begin{algorithmic}
				\caption{{\mstconnect}}\label{algo:mst}
				\STATE \textbf{input} $G=(V,E)$, $L$, $c:L\to\mathbb{R}_{\geq0}$
				\STATE \textbf{output} augmentation $S\subseteq L$
				\STATE \textbf{procedure} {\mstconnect}($G$, $L$, $c$)
				\STATE \quad $C=(V_c, E_c)\leftarrow $\textsc{cactusByVieCut}$(G)$
				\STATE \quad $C_L\leftarrow $\textsc{buildLinkGraph}$(V_c,L)$
				\STATE \quad$L_{MST}\leftarrow$ \textsc{MST}($C_L$, $c$)
				\STATE \quad$S\leftarrow$ \textsc{sortDesc}($L_{MST}$, $c$)
				\STATE \quad\textbf{for} $ l \in {S}$ \textbf{do}
				\STATE \qquad \textbf{if} \textsc{disposableLink}$(l,C,S)$ 
				\STATE \qquad $S\leftarrow S\setminus\{l\}$
				\STATE \quad\textbf{return} $S$
			\end{algorithmic}
		\end{algorithm}	
		\vspace*{-0.85cm}
	\end{minipage}
\end{wrapfigure} 
 The complete set of links can have $\mathcal{O}(|V_c|^2)$ size. Therefore, checking if a link can be removed from a set can also be expensive which is why a small initial solution is essential. Using Theorem~\ref{thm:mst}, an intuitive starting point is a minimum spanning tree (in case of incomplete link sets a minimum spanning forest) $L_{MST}$ of the cactus link graph $C_L$.
For each link in the set $L_{MST}$, beginning with the heaviest, we check whether it can be removed in \textsc{disposableLinks} of Algorithm~\ref{algo:mst}. This process needs linear time, as shown in Theorem~\ref{thm:linearTimeCheck}. If a link is disposable, we exclude it from the solution. Figure~\ref{fig:example} (left) illustrates an example.
\begin{theorem}\label{thm:mst}
	Let ${C=(V_c,E_c)}$ be a cactus graph and ${C_L=(V_c,L_c)}$ the link graph of $C$, such that a feasible solution to the WCAP exists. Then, a minimum spanning forest $L_{MST}\subset L_c$ for $C_L$ is a feasible solution to the WCAP on $C$.
\end{theorem}
\begin{proof}
	If the graph $C_L$ is connected, $L_{MST}$ is a tree. In this case there is a path in $C_L$ between any two vertices and therefore every cut in $C$ is covered. Therefore, $L_{MST}$ is a feasible solution to the weighted connectivity augmentation problem.	
	Assume $C_L$ is not connected and $L_{MST}$ is not a feasible solution to the WCAP, \ie the connectivity of $\tilde{C}=(V_c,E_c\cup L_{MST})$ is not increased compared to the connectivity of $C$. Then, there must be a cut $c$ which is not covered by the links in $L_{MST}$. However, since there exists a solution to the problem there must be a link $l=(u,v)$ in $C_L$ covering the cut $c$. Since $L_{MST}$ is a minimum spanning forest $u$ and $v$ must be in the same connected component. Therefore, there is a path between $u$ and $v$ in $L_{MST}$ which results in $c$ being covered. This forms a contradiction and therefore $L_{MST}$ is a solution to the weighted \hbox{connectivity augmentation problem. }
\end{proof}
\begin{theorem}\label{thm:linearTimeCheck}
	Let ${G=(V,E)}$ be a $k$-connected graph and ${C=(V_c,E_c)}$ its cactus graph representation, $L_{MST}\subset L$ a minimum spanning forest on the link graph $C_L$ and $l\in L_{MST}$. Then, we can check if ${L_{MST}\setminus \{l\}}$ is still an \hbox{augmentation in $\mathcal{O}(|V_c|)$}.
\end{theorem}
\begin{proof}
	A link $l=(u,v)\in L_{MST}$ can be removed, \ie is disposable, if $L_{MST}\setminus \{l\}$ is still an augmentation. 
	Therefore, we check the connectivity between $u$ and $v$ in the graph $C'=(V_c,E_c\cup L_{MST}\setminus\{l\})$ using a maximum flow algorithm. A link $l=(u,v)$ can be removed, if the $u$-$v$-flow and therefore the minimum $u$-$v$-cut in $C'$ is larger than $k$. Let $m=|E_c \cup L_{MST}|$.
	To compute a maximum $u$-$v$-flow we use Ford-Fulkerson with complexity $\mathcal{O}(m\cdot f)$, where $f$ is the flow value~\cite{ford-fulkerson}. 
	We can improve this factor by initially computing augmenting paths only using edges from the cactus graph. Since these edge weights are in $\{\frac{k}{2},k\}$, there can be at most two rounds of augmenting paths, each with weight at least $\frac{k}{2}$. Therefore, this can be done in $\mathcal{O}(|V_c|)$ time.
	Afterwards, to determine whether the $u$-$v$-connectivity in $C'$ is larger than $k$, we only need to do one iteration of Ford-Fulkerson which takes $\mathcal{O}(m)$ time.
	As $C'$ is the union of a cactus and a minimum spanning tree, it holds $m\in \mathcal{O}(|V_c|)$ and we can therefore determine if a link $l=(u,v)$ is \hbox{disposable within $\mathcal{O}(|V_c|)$.}
\end{proof}
A minimum spanning forest on the cactus link graph $C_L$ is computed using Kruskal's algorithm with a complexity of $\mathcal{O}(m\log m)$~\cite{kruskal}, where $m$ is the number of links in $C_L$.
Checking if $\mathcal{O}(|V_c|)$ links can be removed from the solution is done in $\mathcal{O}(|V_c|^2)$, see Theorem~\ref{thm:linearTimeCheck}. This results in an overall complexity of $\mathcal{O}(|V_c|^2\log |V_c|)$ for {\mstconnect}. 

As in {\gwc}, we also tested the cost per augmented minimum cut $c(l)/a_l$ as weight to compute the minimum spanning forest as well as using $c(l)/a_l$ to sort the links to be removed in Algorithm~\ref{algo:mst}.

\begin{wrapfigure}{H}{.5\textwidth}
	\vspace*{-0.85cm}
	\begin{minipage}{.475\textwidth}
		\noindent
		\begin{algorithm}[H]
			\begin{algorithmic}
				\caption{{\ls}}\label{algo:ls}
				\STATE \textbf{input} $G=(V,E)$, $L$, $c:L\to\mathbb{R}_{\geq0}$, $S$
				\STATE \textbf{output} improved augmentation $S'\subseteq L$
				\STATE \textbf{procedure} {\ls}($G$, $L$, $c$, $S$)		
				\STATE \quad $C=(V_c, E_c)\leftarrow $\textsc{cactusByVieCut}$(G)$
				\STATE \quad $L'\leftarrow$ \textsc{reduceLinkSet}($C$, $L$, $c$)
				\STATE \quad $X\leftarrow$ \textsc{getSwapCands}($L'$, $S$)
				\STATE \quad $S'\leftarrow$ $S$
				\STATE \quad \textbf{while} $X\neq\emptyset$ \textbf{do}
				\STATE \quad \quad $(\Lin,\Lout) \leftarrow $ \textsc{bestCandidate}($X$)
				\STATE \quad \quad $X\leftarrow X\setminus(\Lin,\Lout)$
				\STATE \quad \quad \textbf{if} \textsc{isSwapValid}($\Lin$, $\Lout$)
				\STATE \quad \quad \quad $S\leftarrow$ \textsc{swap}($S'$, $\Lin$, $\Lout$)
				\STATE \quad \quad \quad $X\leftarrow$ \textsc{getSwapCands}($L'$, $S'$)
				\STATE \textbf{return} $S'$
			\end{algorithmic}
		\end{algorithm}
		\vspace*{-0.85cm}
	\end{minipage}
\end{wrapfigure}
 
However, our experiments showed that using the weight $c(l)$ of the links performed significantly better in both cases and thus the other versions are omitted.

\textbf{{Local Search}.}
 The core idea of our local search algorithm {\ls} is the following. We want to remove links from our solution and replace them with a lighter set of links. The parameter $k$ limits the number of links within such a swap. Since these swaps can result in infeasible solutions, we have to check feasibility for each swap. This way a non-optimal solution $S\subset L$ can be improved. Next we describe the different steps of {\ls} in detail. We give an overview in Algorithm~\ref{algo:ls}.
We first compute the cactus graph representation $C$ of $G$ using \textsc{VieCut}~\cite{viecut,cactus-state-of-the-art}. Afterwards, we reduce the link set $L$ to the union of $t$ disjoint minimum spanning forests in the link graph $C_L$. Here, the next minimum spanning forest is computed on $C_L$ after removing the links of the previously computed minimum spanning forest. This way the average degree is a small constant depending on $t$. We set $t=2$, since larger values did not \hbox{improve the performance.}

The main idea of the remaining part in {\ls} is the following. We search for non-solution links  $\Lin \subset L\setminus S$ and solution links $\Lout \subset S$ with $|\Lin| \leq k_{in}$, $|\Lout| \leq k_{out}$ and $k_{in} + k_{out} \leq k$. If the swap $(L_{in},L_{out})$ is feasible, we can create a new improved solution by swapping these sets, \ie $S' \coloneqq S \setminus \Lout \cup \Lin$. In \textsc{getSwapCandidates} we compute possible swap candidates $(\Lin,\Lout)$. These have to fulfil the following conditions:
\begin{itemize}
	\item $\Lin \cup \Lout$ form an alternating path from $\Lin$ and $\Lout$ of length at most $k_{in} + k_{out} = k$
	\item $c(\Lin) - c(\Lout) < 0$, \ie the swap improves the solution
	\item non-triviality, \ie $\forall v\in V(\Lout):\{v\}$ is no minimum cut
\end{itemize}

\begin{figure}[t]
	\caption{Example of applying {\mstconnect} (right) and improving its solution with our local search {\lsthree}. All links have weight 1. Green edges represent the current solution, black dashed edges are the edges of the cactus graph representation, light gray edges are the non-solution links. The dashed green ($\Lout$) and red ($\Lin$) edges show the swap found in the local search.}
	\label{fig:example}
	\begin{tikzpicture}
    \definecolor{darkgreen}{rgb}{0.0, 0.5, 0.0}
	\definecolor{lightgreen}{rgb}{0.235, 0.7, 0.235}

	\tikzstyle{mynode} = [circle, fill=black, inner sep=0pt, minimum size=4pt]
	\tikzstyle{circ} = [black, dashed]
	\newcommand{\drawgraph}[2]{
		\foreach \x in {0,45,...,315} {
			\node[mynode] (\x#2) at (\x:1cm) {};
		}
		
		\foreach \x in {0,45,...,315} {
			\foreach \y in {0,45,...,315} {
				\draw [lightgray, thin] (\x#2) -- (\y#2);
			}
		}
					
		\draw [circ] (0#2) -- (45#2) -- (90#2) -- (135#2) --(180#2) --(225#2) --(270#2) --(315#2) --(0#2);
		
		\ifnum#1=1
		\draw [darkgreen, ultra thick] (180#2) -- (0#2);
		\draw [darkgreen, ultra thick] (45#2) -- (135#2);
		\draw [darkgreen, ultra thick] (225#2) -- (315#2);
		\draw [darkgreen, dashed, ultra thick] (270#2) -- (45#2);
		\draw [darkgreen, dashed, ultra thick] (225#2) -- (90#2);
		\draw [red, dashed, ultra thick] (270#2) -- (90#2);
		\fi
		\ifnum#1=2 		
		\draw [darkgreen, ultra  thick] (180#2) -- (0#2);
		\draw [darkgreen, ultra thick] (45#2) -- (135#2);
		\draw [darkgreen, ultra thick] (225#2) -- (315#2);
		\draw [darkgreen, ultra thick] (270#2) -- (90#2);
		\fi
		\ifnum#1=3
		\draw [darkgreen, ultra  thick] (180#2) -- (0#2);
		\draw [darkgreen, ultra  thick] (270#2) -- (90#2);
		\draw [darkgreen, ultra thick] (225#2) -- (135#2);
		\draw [darkgreen, ultra thick] (45#2) -- (180#2);
		\draw [darkgreen, ultra thick] (90#2) -- (315#2);
		\draw [red, ultra thick] (45#2) -- (135#2);
		\draw [red, ultra thick] (225#2) -- (315#2);
		\fi	
		\ifnum#1=4 %
		\draw [darkgreen, ultra thick] (180#2) -- (0#2);
		\draw [darkgreen, ultra thick] (45#2) -- (135#2);
		\draw [darkgreen, ultra thick] (225#2) -- (315#2);	
		\draw [darkgreen, ultra thick] (270#2) -- (45#2);
		\draw [darkgreen, ultra thick] (225#2) -- (90#2);
		\draw [darkgreen, ultra thick] (225#2) -- (45#2);
		\draw [darkgreen, ultra thick] (0#2) -- (270#2);
		\fi
		\ifnum#1=5 %
		\draw [darkgreen, ultra thick] (180#2) -- (0#2);
		\draw [darkgreen, ultra thick] (45#2) -- (135#2);
		\draw [darkgreen, ultra  thick] (225#2) -- (315#2);		
		\draw [darkgreen, ultra thick] (270#2) -- (45#2);
		\draw [darkgreen, ultra thick] (225#2) -- (90#2);
		
		\fi
	}
	
	\begin{scope}[xshift=-9.75cm]
	\drawgraph{4}{E} %
	\node[below=1.2cm] at ($(0E) - (1,0)$) {$L_{MST}$};
	\end{scope}
	\node at (-7.6cm,0.6cm) {drop};	
	\draw [-Stealth] (-8.5cm, 0) -- node[above] {disposable} node[below] {links} (-6.7cm, 0);		
	
	\begin{scope}[xshift=-5.4cm]
		\drawgraph{5}{D} %
	\end{scope}
	\node[below=1.2cm] at ($(0D) - (1,0)$) {\mstconnect};
	\node[below=1.7cm] at ($(0D) - (1,0)$) {Solution};
	
	\draw [-Stealth] (-4.1cm, 0) -- node[above] {improve} node[below] {} (-2.8cm, 0);
	
	\begin{scope}[xshift=-1.5cm]
		\drawgraph{1}{A} %
		\node[below=1.2cm] at ($(0A) - (1,0)$) {\lsthree};
		\node[below=1.7cm] at ($(0A) - (1,0)$) {Step};
	\end{scope}

	\draw [-Stealth] (-0.3cm, 0) -- node[above] {swap} (.7cm, 0);
	
	\begin{scope}[xshift=2cm]
		\drawgraph{2}{B} %
		\node[below=1.2cm] at ($(0B) - (1,0)$) {Optimal};
		\node[below=1.7cm] at ($(0B) - (1,0)$) {Solution};
	\end{scope}
\end{tikzpicture}
\end{figure}
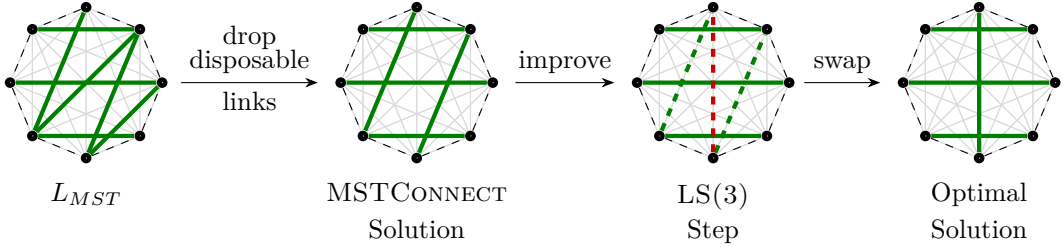
	We restrict the search for swaps to alternating paths. Since endpoints of these paths in $\Lout$ lose an adjacent link, it is likely that a new adjacent link is necessary to keep the cut covered. Interior links of the path cover two vertices at once.
	Utilizing the specified criteria, we determine the sets $\Lin$ and $\Lout$ within the \textsc{getSwapCandidates} function. This is achieved through an adapted depth-first search process, where traversal is restricted to alternate edges that are part of the solution and those that are not. The depth of the search is restricted by the parameter $k$. This modified search is executed starting at each vertex in $|V_c|$. At every vertex encountered, we verify the adherence to the second and third condition. If they are met, the path currently identified is included in the collection of potential swap candidates.
	These conditions are not sufficient to form a feasible swap, they rather help to prune the set of possible candidates in advance. For the best candidate, \ie the candidate where $c(\Lin) - c(\Lout)$ is lowest, we check the validity of the swap in \textsc{isSwapValid}. Using a maximum flow approach, similar to {\mstconnect}, this check can be done efficiently as shown in Theorem~\ref{thm:linearTimeCheck}. If the swap is valid, we swap links, recompute possible candidates and start over. If the swap is not valid, we check the next best candidate. Figure~\ref{fig:example} shows an example of how the $\lsthree$ algorithm improves the solution computed by {\mstconnect}. We further improved the algorithm {\ls} by integrating path caching, where we maintain vertex sequences of a path already checked using a hash table. 	Using Theorem~\ref{thm:ls}, the running time for {\lsthree} is $\mathcal{O}(|V_c|^3\times \Delta^2 + |V_c|^4)$ and for {\lsfive} $\mathcal{O}(|V_c|^4\times \Delta^2 + |V_c|^5)$. 
\begin{theorem}\label{thm:ls}
	{\ls} has a complexity of $\mathcal{O}(|V_c|^2(|V_c|^{\lfloor\frac{k}{2}\rfloor}\times \Delta^2 + |V_c|^{\lceil\frac{k}{2}\rceil}))$.
\end{theorem}
\begin{proof}
	The running time for {\ls} is dominated by the while loop in Algorithm~\ref{algo:ls}. We have at most $|V_c|$ re-computations of the set $X$. Therefore, we need to check validity for at most $|V_c|\times|X|$ swaps. This results in a running time of $\mathbb{O}(|V_c|^2\times |X|)$. 
	To estimate $|X|$ we only consider the swaps using exactly $k$ links, since these are the dominating factor. Each path of length exactly $k$ has to have at least $k_1= \lfloor\frac{k}{2}\rfloor$ and at most $k_2 = \lceil\frac{k}{2}\rceil$ edges from the current solution $S$. There can be $\mathcal{O}(|V_c|^{k_1}+ |V_c|^{k_2})$ swaps in $X$ using different links from $S$. Now we estimate the number of different paths using the same links from $S$. Let $k$ be odd and there are $k_2$ equal solution links in the path. Then, there can only be one such path, since the endpoints are fixed. If there are only $k_1$ equal solution links in the path, then both endpoints can vary. This results in at most $\Delta^2$ different paths, where $\Delta$ is the max degree in $(V_c,L'\cup S)$. When $k$ is even, one endpoint of the path is free and there are at most $\Delta$ different paths. Overall, this yields $|X| \in \mathcal{O}(|V_c|^{\lfloor\frac{k}{2}\rfloor}\times \Delta^2 + |V_c|^{\lceil\frac{k}{2}\rceil})$ and results in a complexity for {\ls} of $\mathcal{O}(|V_c|^2(|V_c|^{\lfloor\frac{k}{2}\rfloor}\times \Delta^2 + |V_c|^{\lceil\frac{k}{2}\rceil}))$. 
\end{proof}

\subsection{Exact Algorithm}
We now describe an efficient implementation of an integer linear program for the problem, which can be used to solve (small) instances to optimality. The formulation is inspired by the linear program used for the 2-approximation in~\cite{tap}. To the best of our knowledge, this formulation itself has not been used before. %
For the formulation of the WCAP, we introduce binary variables $x\in L$ that decide which links are added to the augmentation. The objective given in~\eqref{eq:ilp} sums up the weight of all selected links. The constraint ensures that each minimum cut in the graph $G$ is increased by at least one link added to the augmentation.
\begin{equation}
	\min_{x}\sum_{l\in L}x_lc(l) \quad
	\text{s.t.}\quad \sum_{l\in L}\operatorname{cut}(c, l)x_l\geq1\quad \forall c\in C_G;  \quad
	x_l\in\{0,1\}^{|L|}. 
\label{eq:ilp}
\end{equation}

\vspace*{-.2cm} 
Our ILP formulation is run on the cactus representation and uses efficient enumeration algorithms on it to list all minimum cuts \hbox{$C_G$ of \textsc{VieCut}~\cite{cactus-state-of-the-art}}. Moreover, our solver, which we call {\ILP}, uses an initial solution computed by {\mstconnect}, which improves the running time by \numprint{5.5}\% on cycle graphs, which are the most difficult instances to solve with {\ILP}. %
We choose {\mstconnect} since it is the fastest heuristic approach already giving good \hbox{results (see Section~\ref{sec:experiments})}.

\section{Experimental Evaluation} \label{sec:experiments}

 We now experimentally evaluate the algorithms described in the previous sections. The approximation algorithms are evaluated and compared in terms of quality and running time in Section~\ref{sec:ex:apx}. Afterwards, our new proposed algorithms {\gwc}, {\mstconnect} and {\mstconnect}+{\ls} are compared to the state-of-the-art solvers {\HBD}, {\FSM} and {\SMC} in terms of solution quality, running time and memory consumption. 
 
 \textbf{Methodology.} %
 The experiments are run on a computer with an AMD~EPYC~9754~128~core processor with 256 threads running at \numprint{3.1}\,GHz and \numprint{755}\,GB of main memory running Linux. The C++ code is compiled using \lstinline|gcc 11.4.0| with optimization level \lstinline|O3|. The memory for each process is limited to \numprint{50}\,GB and the running time is limited to 3 hours. 
 We run our algorithms on each instance with 5 different seeds to generate link costs as described below or in cases
 of generated cactus graphs we create 5 different graphs with the same number of vertices and cycles.  %
 The objective, \ie the weight of the augmentation, the running time and the maximum memory used is measured. We use the geometric mean when averaging over different seeds or instances such that every instance has a comparable influence on the result.
 
 Different algorithms are compared using performance profiles~\cite{perf-prof}. These plots use the best algorithm as baseline for each instance and relate the other algorithms to this baseline. A performance profile can use the objective function to compare quality, running time and memory consumption. The $x$-axis shows a parameter $\tau\geq1$. On the $y$-axis the fraction of instances whose objective is at most $\tau\cdot\mathrm{best}$ is plotted, in particular $\#\{\mathrm{objective}\leq\tau\cdot\mathrm{best}\}/\#\mathrm{instances}$. For running time and memory usage, time and memory are used instead of the objective, respectively. At $\tau=1$ the plot shows the fraction of instances where the algorithm is able to find the best solution / has the fastest running time or lowest memory consumption.
 Some algorithms are not able to solve every instance due to constraints on memory and time. We give more details in the respective sections.
 To solve the integer linear program in {\ILP} and the linear program in the 2-approximation algorithm we used the Gurobi Optimizer~\cite{gurobi}. The integer linear program solver uses presolving and the {\mstconnect} solution as initial solution.

\textbf{Instances.} 
 The algorithms are evaluated using two different sets. The first set consists of three types of \emph{generated graph instances}: cycles or ring graphs, stars and cactus graphs. Cycles and stars represent edge cases of cactus graphs, with an amount of minimum cuts between $\mathcal{O}(|V_c|)$ and $\mathcal{O}(|V_c|^2)$. 
 To be able to test the algorithms on instances that represent more complex and larger cactus graphs, we generated cactus graphs by the algorithm described in Appendix~\ref{sec:generator}. All instances are listed in Table~\ref{tab:apx:instances} including their properties.
 The second set are \emph{real-world instances}. 
 Many real-world graphs have unique or very few distinct minimum cuts, which leads to very small cactus graphs with only a few vertices. 
 We used all connected graphs with non-trivial cactus graph representations from the 10th DIMACS Implementation Challenge~\cite{dimacs} for which the cactus graph representation has at least \numprint{100} edges and at most \numprint{40000} vertices. 
In general, we use \textsc{VieCut}~\cite{cactus-state-of-the-art} to compute the cactus representation of a graph. We only compute the cactus representation once for every graph and henceforth only report running times of the algorithms when run on the cactus representation.

\textbf{Link Sets.}
  The algorithms can perform differently based on the distribution of the link costs. 
  For instance the performance of {\mstconnect} is affected due to different structured minimum spanning trees and lower cardinality of optimum solutions make matching-based approaches like {\FSM} and {\HBD} or our greedy heuristic {\gwc} more efficient. 
Thus, we choose a uniform distribution in different intervals. First, we have a set of \textit{small link cost} choosing the costs uniformly from the sets $\{1,2\}$, $\{1,\dots,9\}$ and $\{1,\dots,99\}$ as used in \cite{watanabe-heuristics,watanabe-matching} to reproduce their results. The second set of \textit{large link costs} consists of links chosen uniformly from the set $\{1, \dots,  \numprint{100000}\}$.
Since our algorithms need link costs within the interval $[0,1]$, we scale those costs  by dividing through the largest occurring link cost. 
  Furthermore, all instances have complete link sets, \ie $G=(V,E\cup L)$ is a complete graph. 
  
 \subsection{Approximations}\label{sec:ex:apx}
 \begin{figure}[t]
		\centering
		\caption{Performance profile comparing solution quality of approximations with $\epsilon=0.15$ on tiny cycle and star instances on the left. The right plot shows the running times.}
		\label{fig:ex:apx}
                \vspace*{-.5cm}
		\includegraphics[width=.49\textwidth]{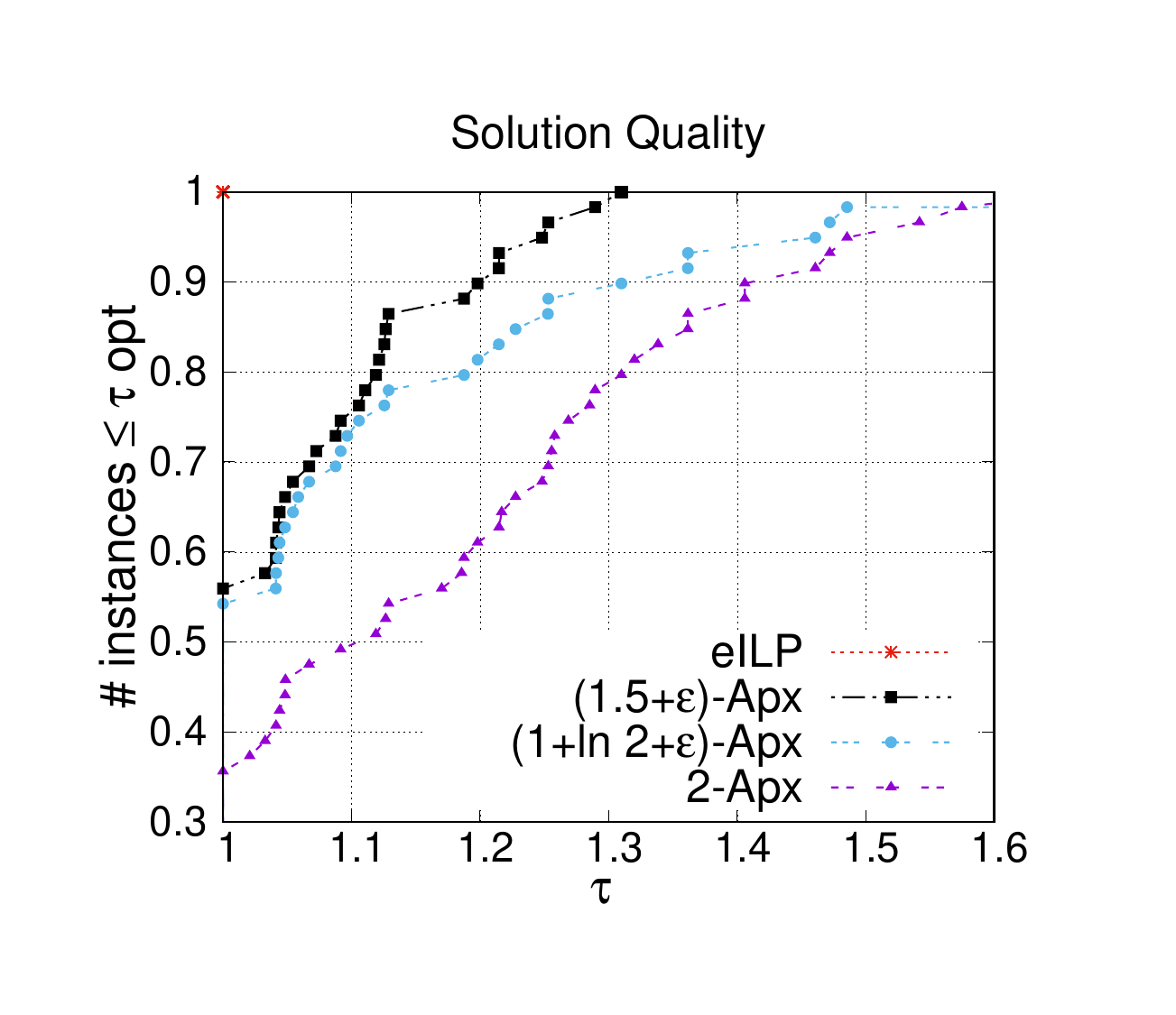} 
		\includegraphics[width=.49\textwidth]{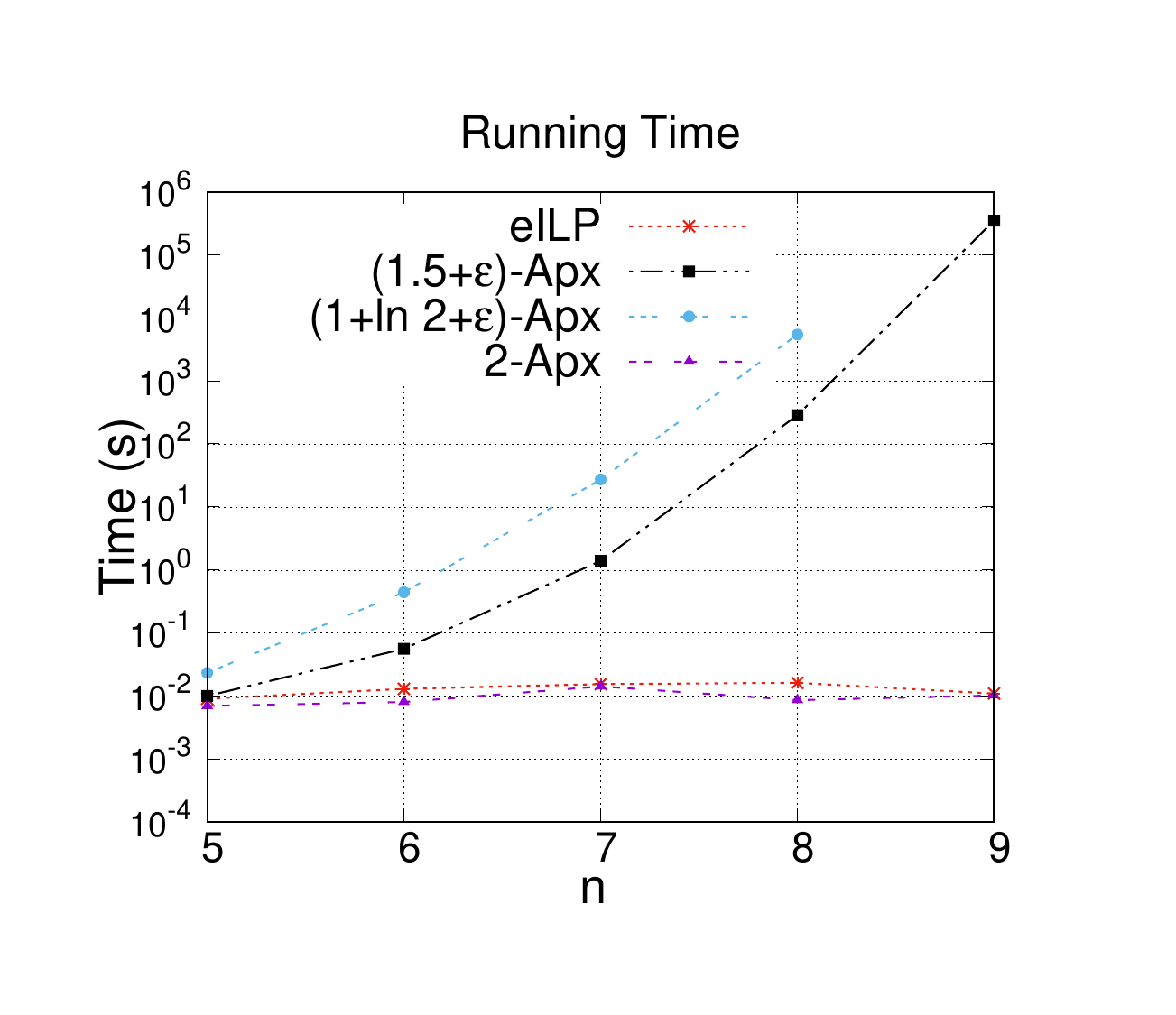} 
                \vspace*{-1.cm}
\end{figure}

We first analyse the performance of the different approximation algorithms described in Section~\ref{sec:apx}. Here, we can only consider very small graphs, as the approximation algorithms (see below) do not scale at all. 
For following experiments, we only report the results for $\epsilon=0.15$.
This is due to the fact that $\epsilon$ is only used to limit the number of considered links crossing cuts to $4\lceil \frac{4}{\epsilon} \rceil$ for the $(1.5+\epsilon)$-Approximation and $4\lceil \frac{2}{\epsilon} \rceil$ for the $(1+\ln2+\epsilon)$-Approximation. 
However, on small graphs the number of links that may cross a minimum cut is already limited by the number of vertices.
Thus, although we tested values for $\epsilon$ in $[0.1, 0.5]$, it did not yield different result in terms of running time or link cost.  

Figure~\ref{fig:ex:apx} (left) shows a performance profile of the solution quality, for tiny graphs. As the optimal solution was computed by {\ILP}, the $x$-axis gives the approximation ratio. The 2-approximation consistently gives the worst solutions, which means both, the $(1+\ln2+\epsilon)$-approximation and the $(1.5+\epsilon)$-approximation, can improve this initial solution. 
 
 Figure~\ref{fig:ex:apx} (right) shows the running time of the approximation algorithms as well as our solver {\ILP} with respect to the graph size. The graphs are cycle graphs with a complete set of links. The running time increases exponentially for the approximation algorithms. This is expected for graphs with fewer than $\alpha\geq28$ vertices as the computational complexity is exponential in~$4\min(n,\alpha)+7$.
 Both, the $(1+\ln2+\epsilon)$-approximation and the $(1.5+\epsilon)$-approximation are orders of magnitude slower than the optimal integer linear program. Hence, we do not consider them further when comparing against other state-of-the-art algorithms. 
 The 2-approximation is slightly faster than {\ILP}, but the difference is negligible and both can easily solve tiny graphs solvable by the dynamic program based approximations. 
 Overall, we \emph{conclude} that the approximations have very little relevance for solving connectivity augmentation problems and may only have theoretical value.
 Since the 2-approximation is able to solve larger graphs it is also used for the comparisons in the following. 

 \subsection{State-of-the-Art Comparison}\label{sec:ex:algorithms}
 In this section we compare our algorithms {\gwc}, {\mstconnect} and {\mstconnect} combined with {\ls} as well as {\ILP} against the performance of the 2-approximation algorithm as well as to the best,~i.e.~{\FSM} and {\HBD}, and the fastest,~i.e.~{\SMC}, state-of-the-art solvers from Watanabe~et~al.~\cite{watanabe-heuristics,watanabe-heuristic-2,watanabe-matching}.
 For the algorithms from Watanabe~et~al.~\cite{watanabe-heuristics,watanabe-heuristic-2,watanabe-matching}, neither the instances used in their experimental evaluation nor source code or binaries for the algorithms are available.\footnote{We contacted the authors, however, we did not get an answer.} Hence, we compare them against \hbox{\emph{our} implementation of these algorithms.} 
For the comparison of solution quality, running time and memory consumption we give performance profiles for generated and real-world instances in Figure~\ref{fig:ex:smalllink} (small link cost) and Figure~\ref{fig:ex:largelink} (large link cost).

\emph{General Remarks.} First note that the 2-approximation algorithm yields the worst results (highest cost) on both data sets and link costs. 
Additionally, it is also the slowest and most memory consuming algorithm. The high memory consumption results from a) the reduction of undirected to directed links, which doubles the size of the link set which roughly gives factor of two and b) more importantly from running the linear program solver. 
We now analyse results for the remaining algorithms for different link costs and \hbox{instances types separately.}

  	\begin{figure}[!t]
  		\centering
    \caption{Performance profile for the state-of-the-art comparison on solution quality, running time and memory consumption on instances with \emph{small cost links}.}
  		\label{fig:ex:smalllink}
                  \includegraphics[width=.8\textwidth]{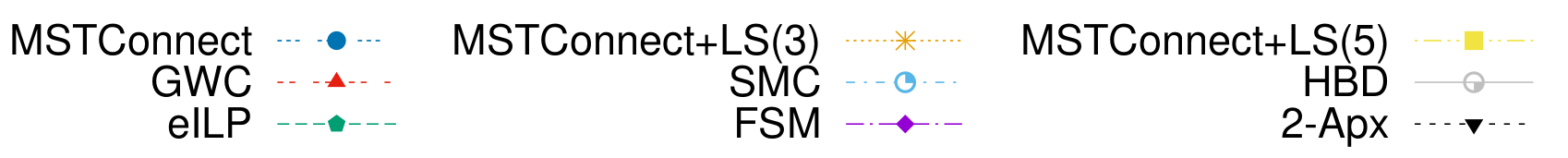}
                  \includegraphics[width=.49\textwidth]{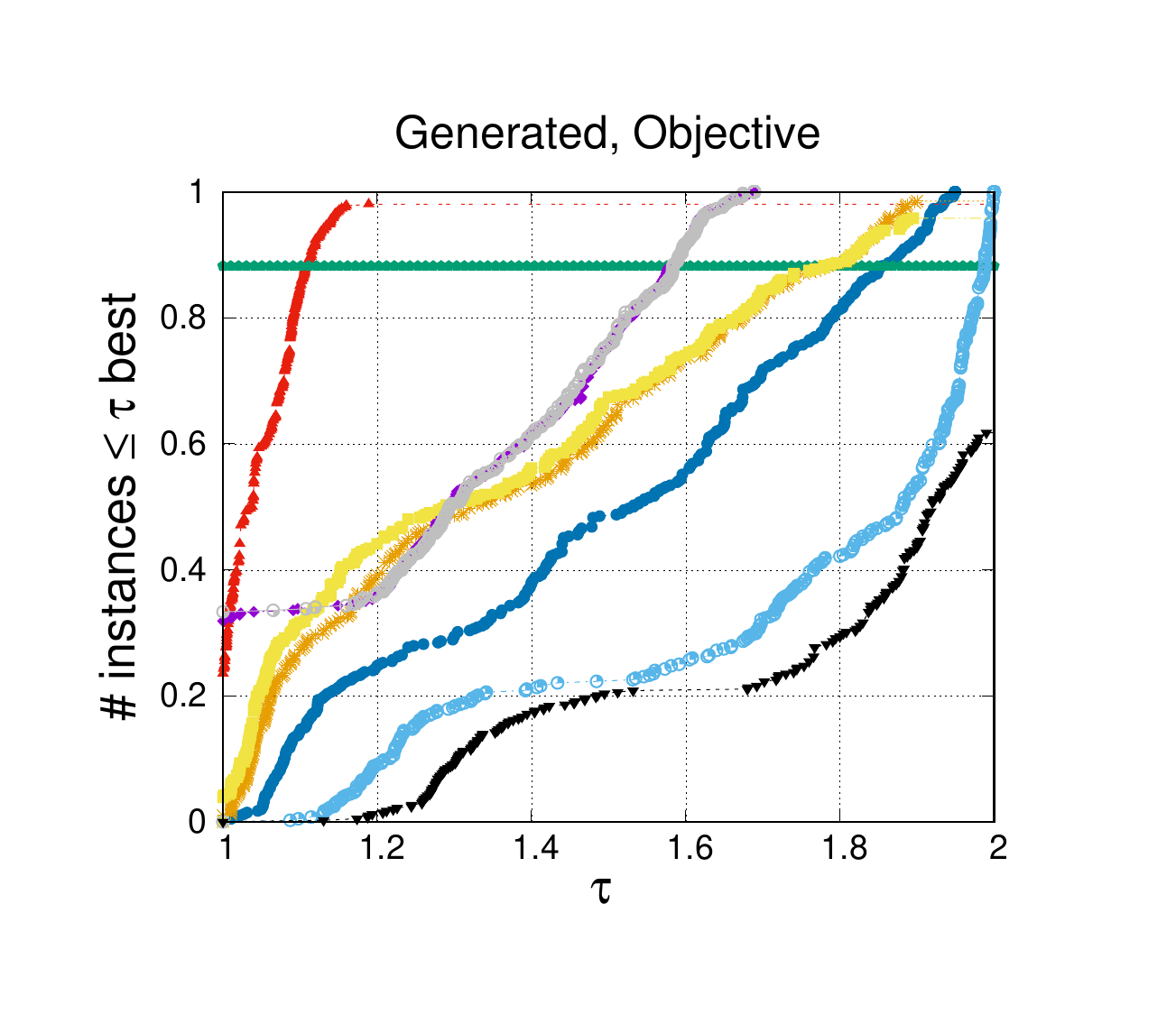} 
                  \includegraphics[width=.49\textwidth]{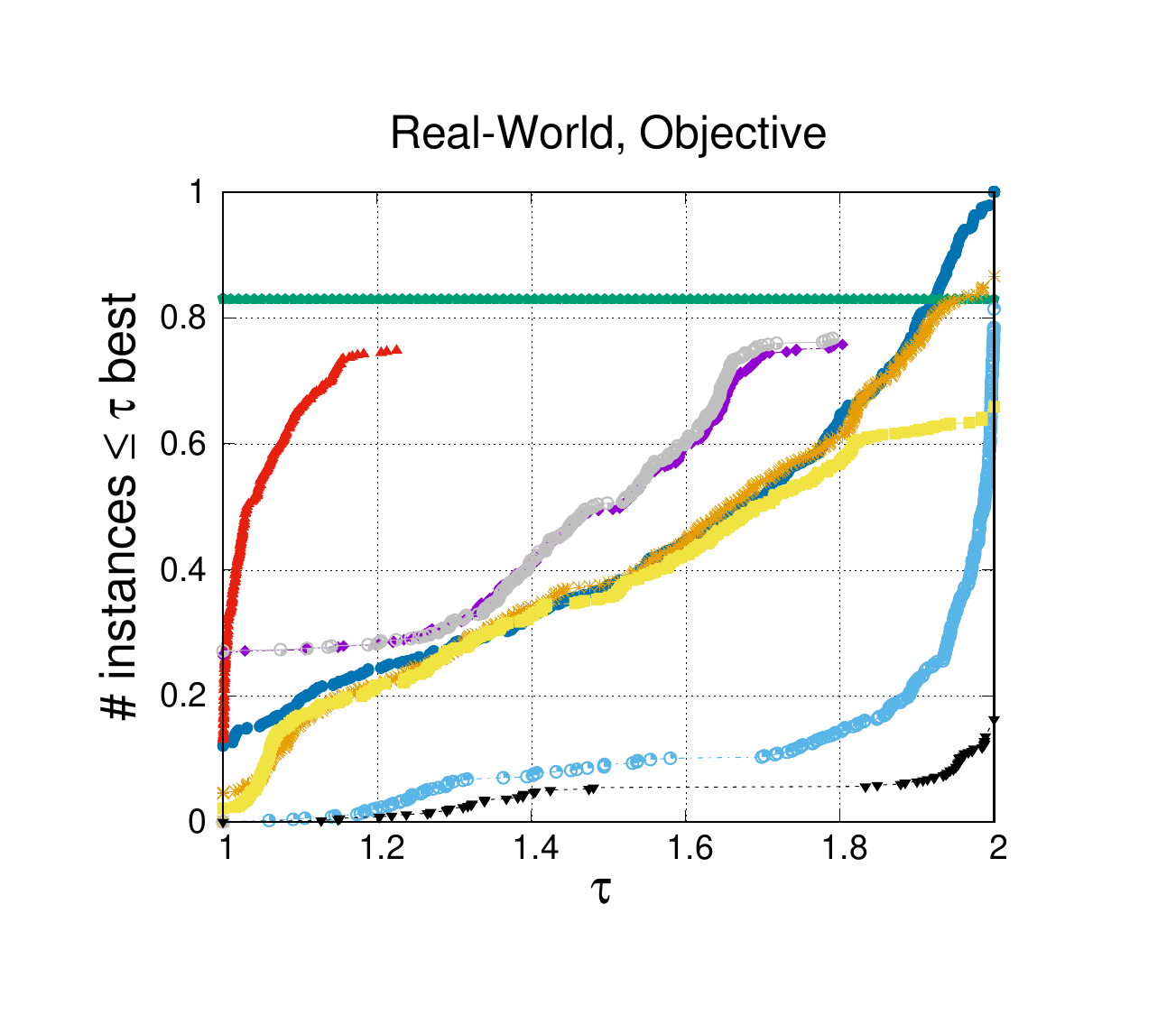}  \\
                                         \vspace*{-1cm}
                  \includegraphics[width=.49\textwidth]{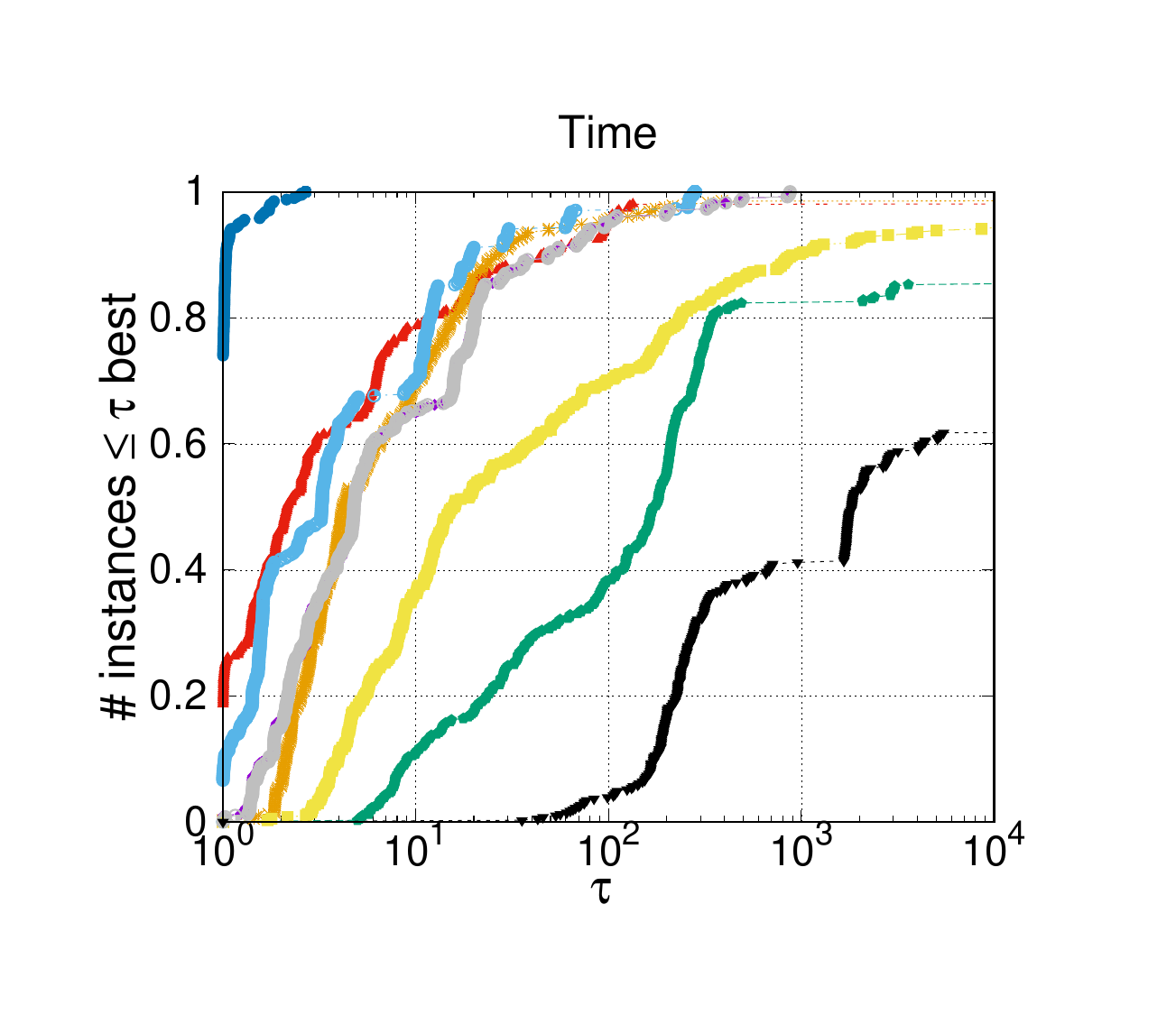} 
                  \includegraphics[width=.49\textwidth]{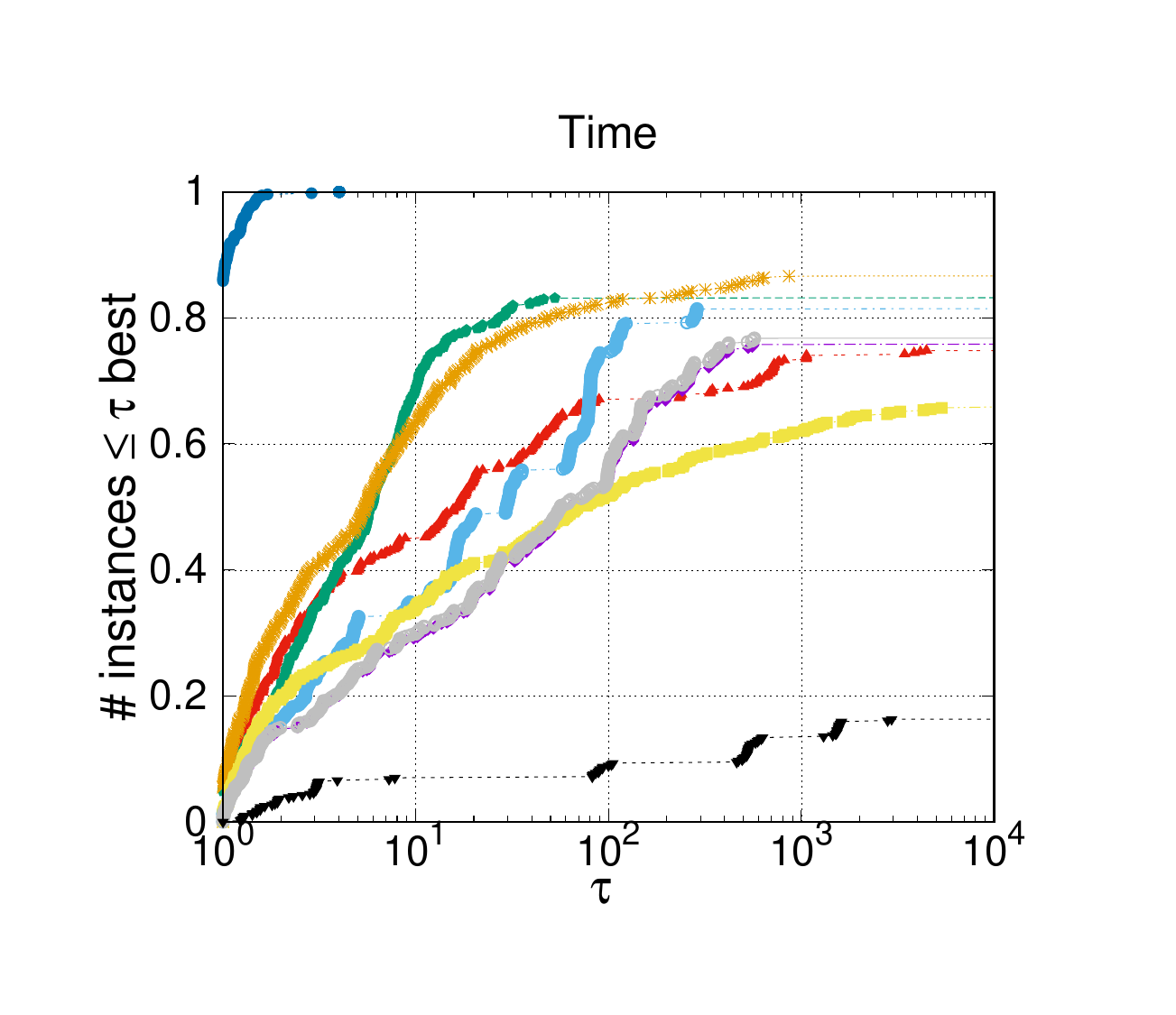} \\
                                         \vspace*{-1cm}
                  \includegraphics[width=.49\textwidth]{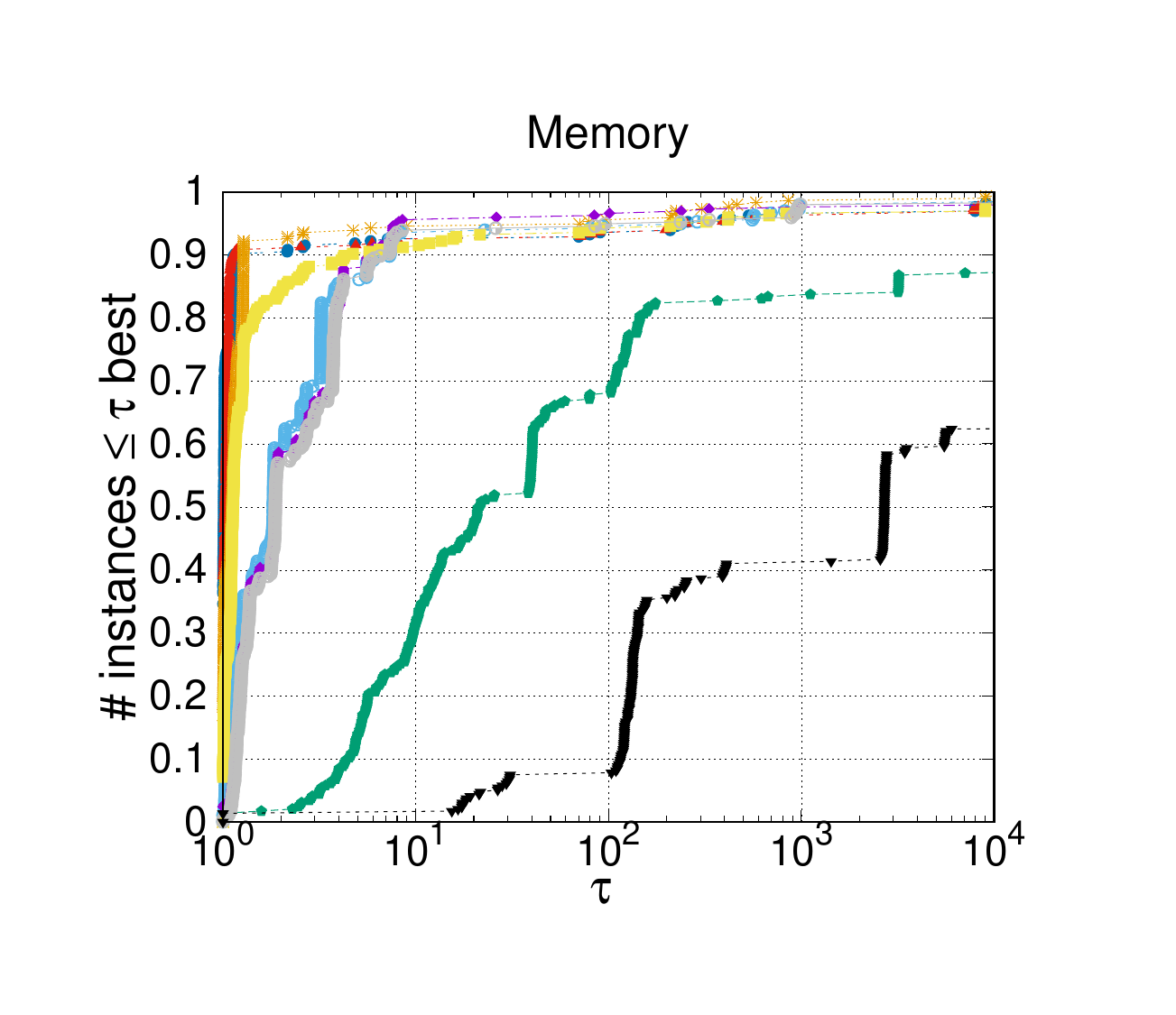}  
                  \includegraphics[width=.49\textwidth]{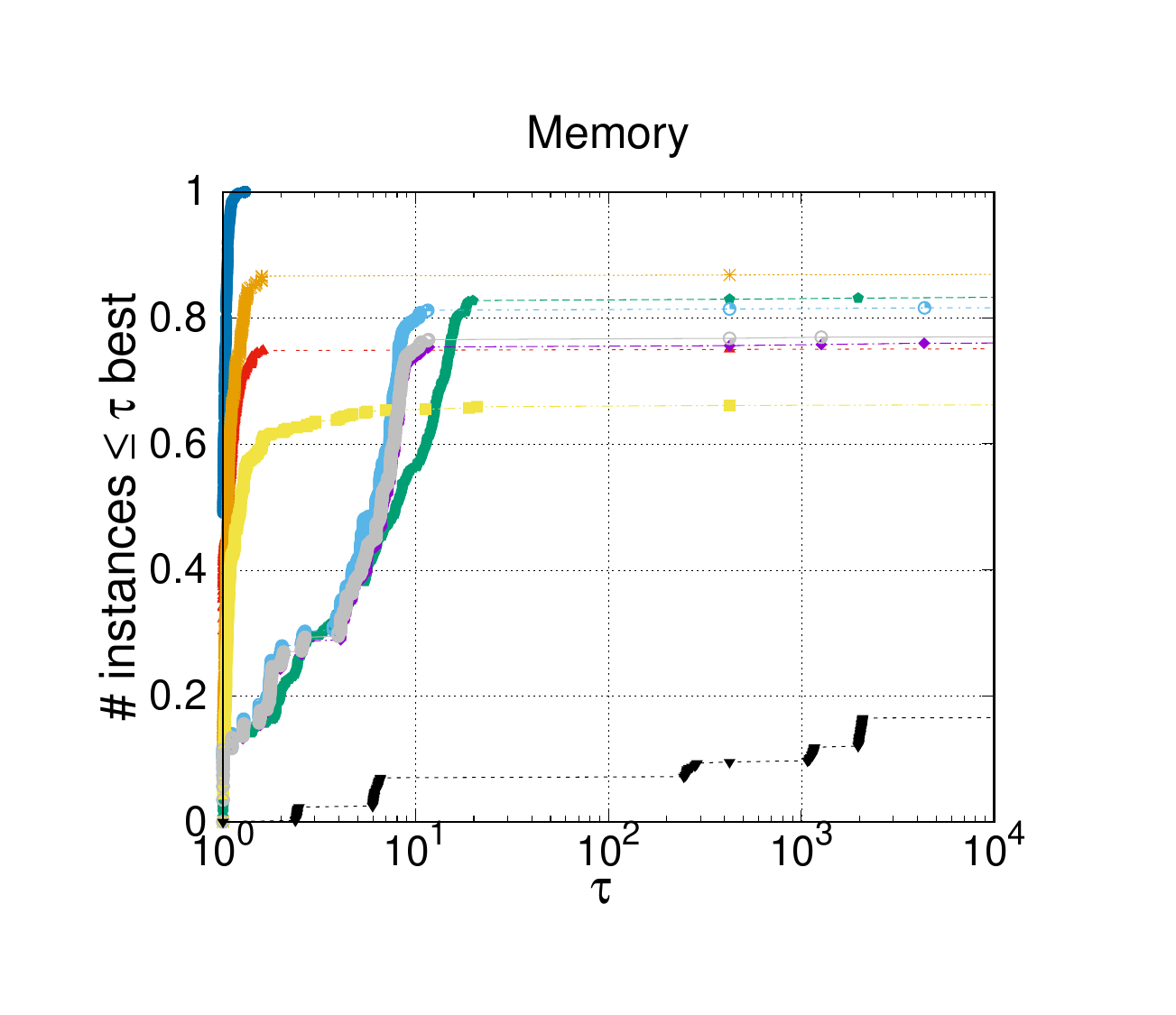} \\
                                         \vspace*{-1cm}

  	\end{figure}

  	\begin{figure}[!t]
  		\centering
        \caption{Performance profile for the state-of-the-art comparison on solution quality, running time and memory consumption on instances with \emph{large cost links}.}
  		\label{fig:ex:largelink}
\includegraphics[width=.8\textwidth]{add_plots/legende.png}
                  \includegraphics[width=.49\textwidth]{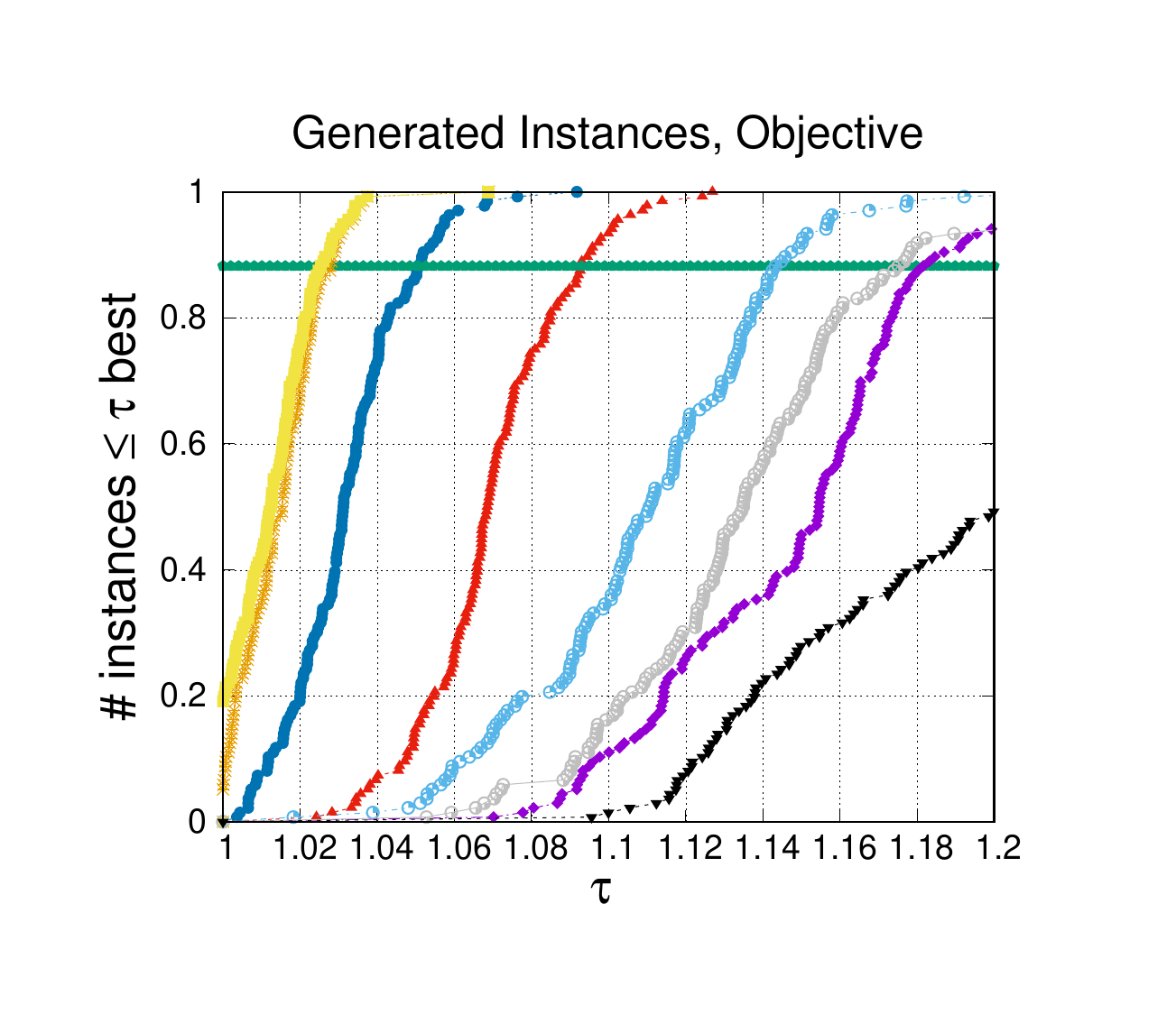} 
                  \includegraphics[width=.49\textwidth]{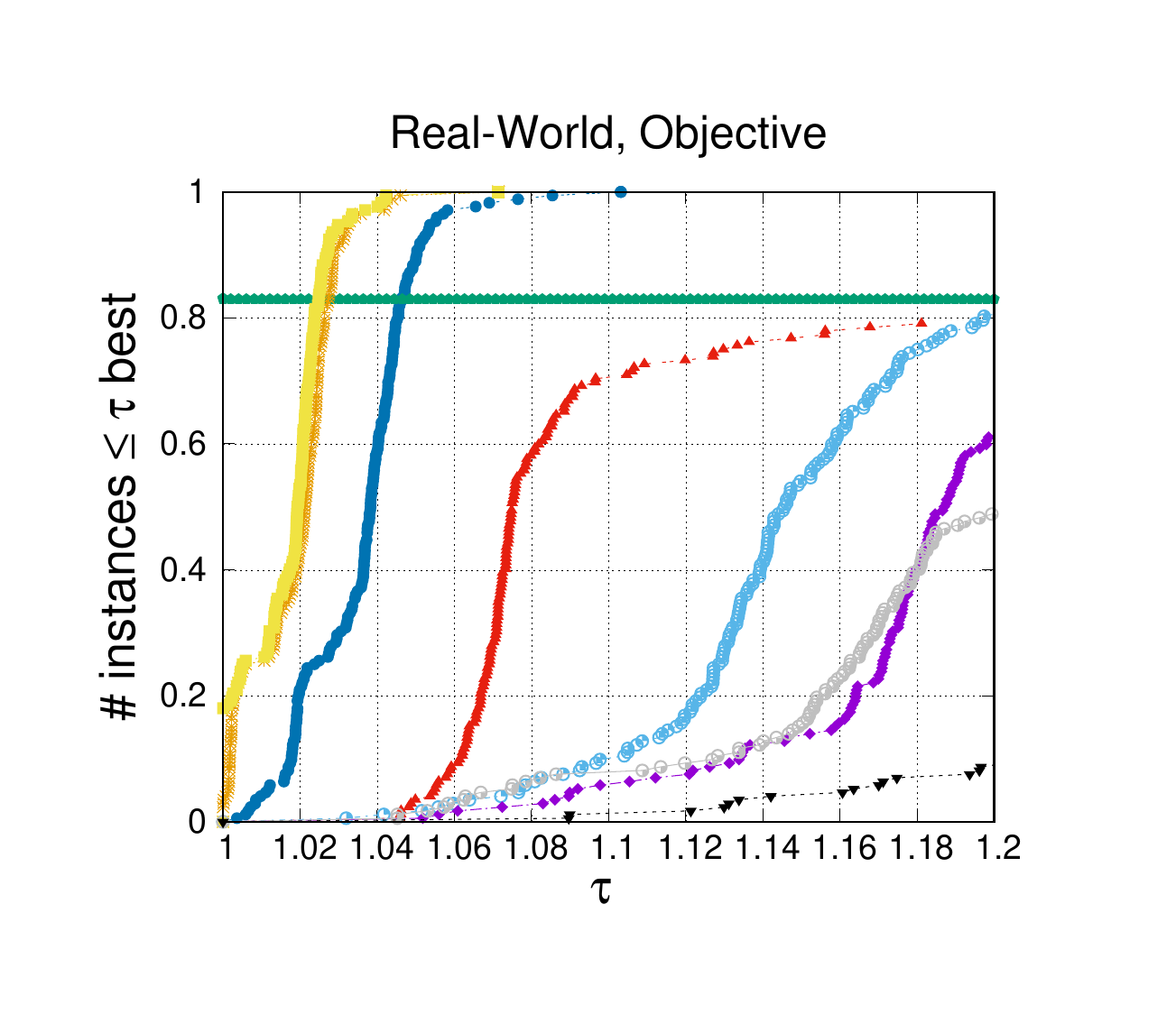} \\
                                         \vspace*{-1cm}
                  \includegraphics[width=.49\textwidth]{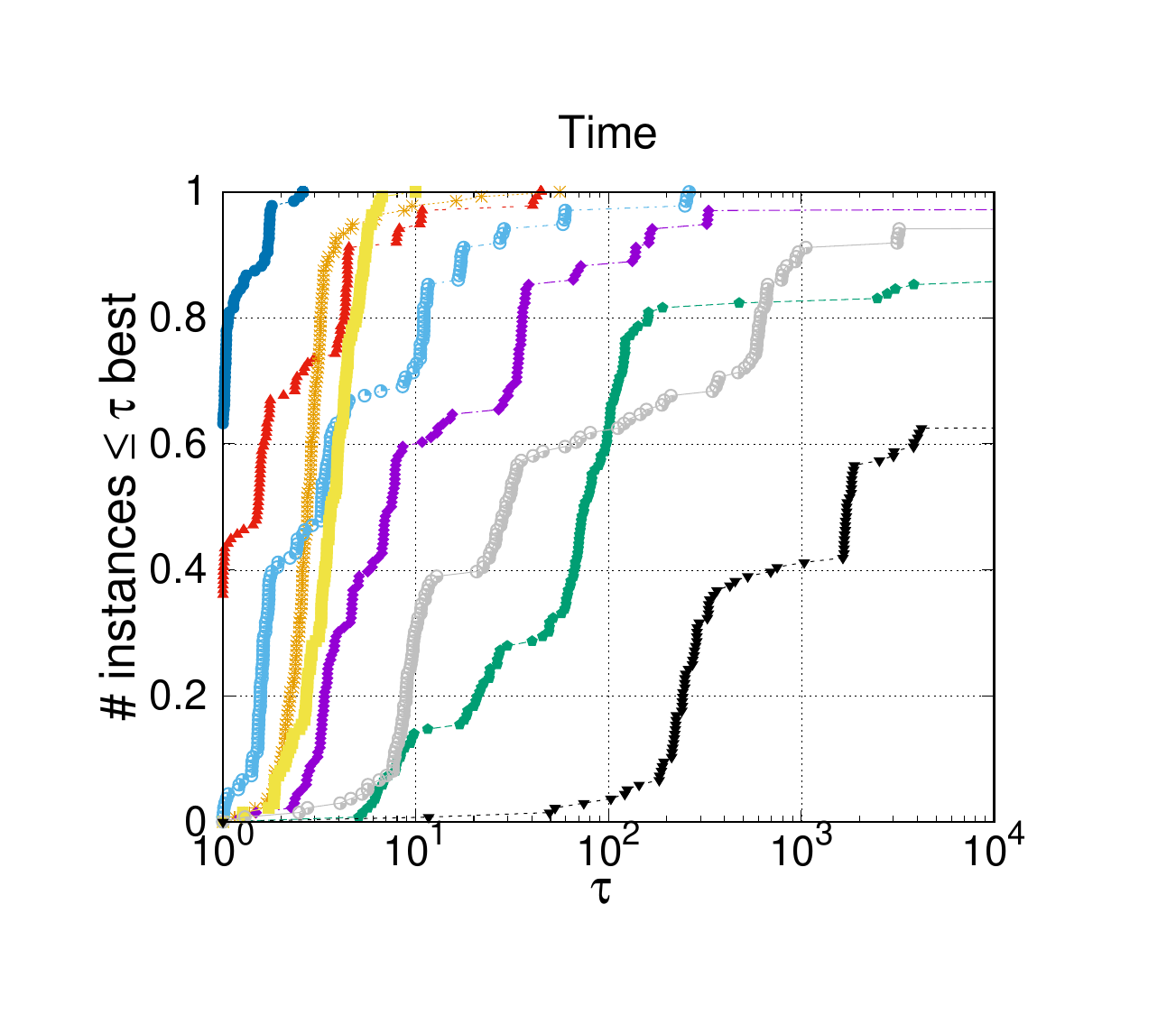} 
                  \includegraphics[width=.49\textwidth]{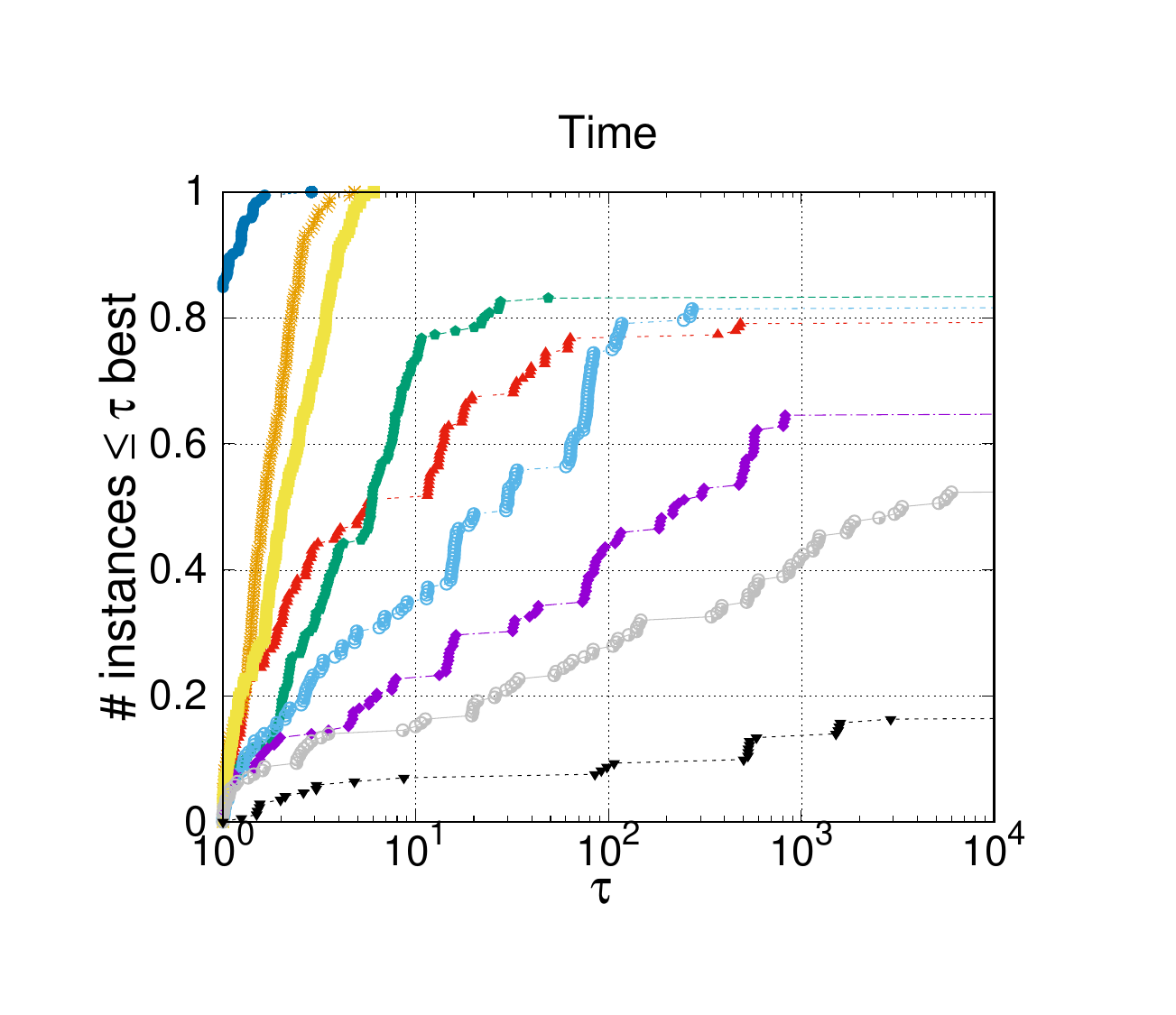} \\
                                         \vspace*{-1cm}
                  \includegraphics[width=.49\textwidth]{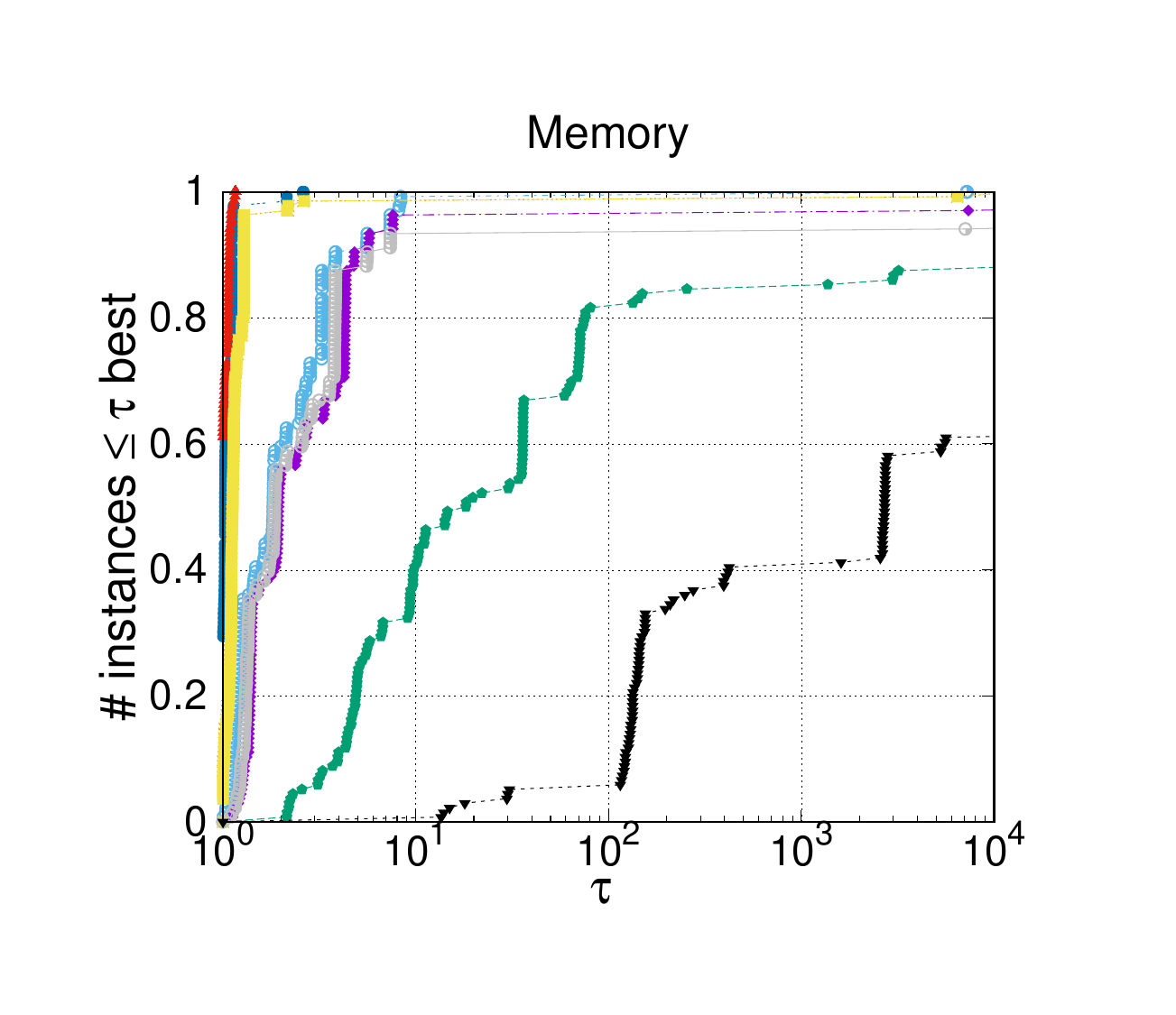}  
                  \includegraphics[width=.49\textwidth]{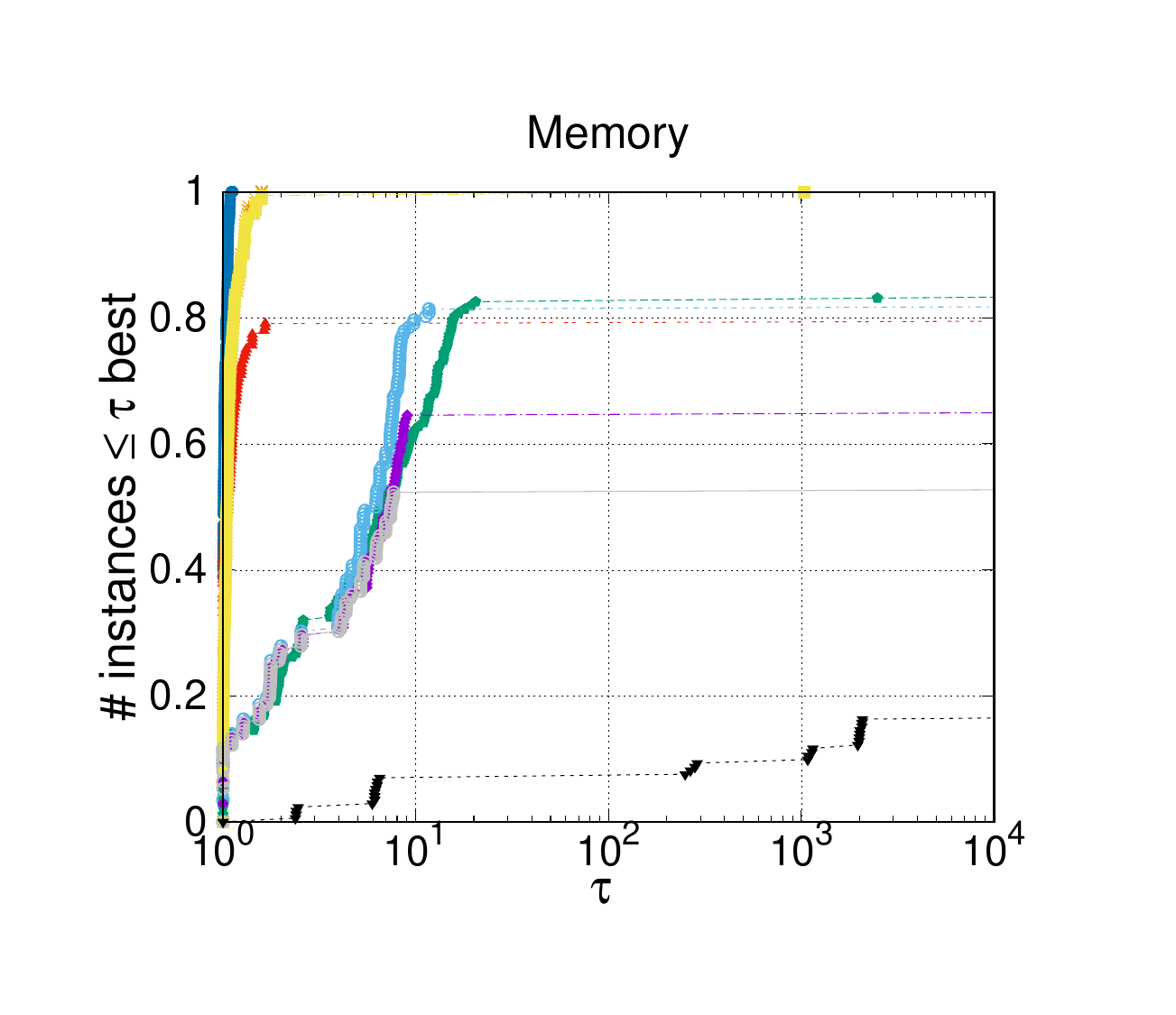} \\
                                         \vspace*{-1cm}
                  
  	\end{figure}

\textbf{Small Link Costs.}
  	The best performing \emph{heuristic} in terms of solution quality on the graphs with small link costs is {\gwc}, see Figure~\ref{fig:ex:smalllink}. Unlike the other algorithms, {\gwc} uses the cost per augmented minimum cut to decide for the solution links. For small link costs, there are a lot of links having the same or similar costs. Therefore, the cost per augmented minimum cut is more impactful than relying only on the cost of a link.
  	For small link costs the {\mstconnect} algorithms are outperformed by the competitors {\FSM} and {\HBD} regarding the solution quality.
    On the instances solvable by all algorithms, {\gwc} achieves solutions that are on average only 4\% larger than the optimum solution, while 
	the results computed by the competitors {\FSM} and {\HBD} are on average \numprint{28}\,\% larger than the optimum solution value. 
  	The {\SMC} algorithm computes worse results than any of our algorithms, performing only slightly better than the 2-approximation.
    Overall, our exact solver {\ILP} can solve \numprint{85.6}\,\% of these instances to optimality. 
 	As can be seen in Figure~\ref{fig:ex:smalllink} (middle), 
    {\mstconnect} is the overall fastest algorithm.  On \emph{generated} instances with small link costs {\gwc}, which yielded the overall best results on these graphs, is on average the second fastest algorithm for these instances. 
    On \emph{real-world} instances on the other hand, our exact solver {\ILP} is on average more than \numprint{3} times faster than {\gwc} and also outperforms all the previous state-of-the-art competitors regarding running time. Since the cactus graph representation of real-world instances usually do not have large cycles, this benefits our solver {\ILP}. The number of minimum cuts and therefore the number of constraints of the ILP~\eqref{eq:ilp} grows quadratically with the size of the largest cycle and only linearly with tree edges.

\textbf{Large Link Costs.}
	The best (heuristic) results for graphs with large link costs are achieved by {\mstconnect} 
    improved by {\ls} for $k=5$ closely followed by $k=3$, see Figure~\ref{fig:ex:largelink} (top). On this set, the wide range of different link costs makes the cost per link more significant, which benefits {\mstconnect}.
    \emph{All} competing algorithms {\HBD}, {\FSM} and {\SMC} yield worse results compared to all our {\mstconnect} approaches 
    including {\mstconnect} without local search and {\gwc}. 
   In particular, our fastest algorithm {\mstconnect}, is on average \numprint{8}\,\% (restricted to the instances solvable by {\SMC}) %
    better compared to {\SMC}, the best performing competitor on these instances, while also being 7 times faster.
   When comparing to the exact solution on the instances solvable by {\ILP}, {\mstconnect} is on average \numprint{3.8}\,\% away from the optimal solution. When additionally using our local search to improve the solution,~i.e. using LS($5$), the result is only \numprint{1.8}\,\% away from the optimum solution.
    All of our heuristic algorithms outperform even the fastest competitor {\SMC} 
    with regard to running time on most large cost real-world instances. On all instances even our slowest algorithm {\gwc} is \numprint{1.7} %
    times faster than {\SMC}, while our fastest algorithm {\mstconnect} outperforms {\SMC} by a factor of \numprint{6.3}.~%
    When only considering real-world instances, the difference gets even larger.
    Here, {\mstconnect} is on average \numprint{9.1} %
    times faster than {\SMC}
    and even our exact solver {\ILP} is a factor of \numprint{2.6} times faster than {\SMC}.
    \newpage
	
	\textbf{Memory Consumption.}
{\mstconnect} requires overall 
    the least amount of memory. Indeed, improving the solution with {\ls} increases the memory consumption with increasing $k$, 
    especially on small link cost instances. Still, the memory usage for {\mstconnect} + {\lsfive} is on 
    average \numprint{1.8} times better than the memory consumption of the previous state-of-the-art algorithms. %
 	On average all our algorithms outperform the competitors regarding memory consumption on \emph{generated} graphs. For the \emph{real-world} instances
    only {\mstconnect} and {\mstconnect} + {\lsthree} (with large link costs also {\lsfive}), 
    need less memory than the other algorithms on almost every instance. Overall, {\mstconnect} uses \numprint{2.6} times less memory than 
    the competitor {\SMC} with the lowest memory consumption. Lastly, we like to note that most of the real-world instances not solvable by {\ILP} are only solvable with our {\mstconnect} approaches.

\textbf{Advise for Practitioners.} 
 A critical factor for the algorithms is the cost associated with the links of the augmentation problem. When dealing with scenarios where the link costs are very large, it is advisable to lean towards \mstconnect~algorithms. 
Conversely, for cases where the link costs are relatively small and instances are similar to the generated ones used in this work, the \gwc~algorithm emerges as a preferable choice. 
For other real-world instances with small link cost, the \ILP~method is the algorithm of choice. In general, local search can help to improve the result at the expense of running time.

\section{Conclusion}
In recent years, new scalable algorithms for the minimum cut and the all minimum cut problem have been engineered~\cite{mincut,mincutjv}.
These algorithms are important subroutines for algorithms that tackle the connectivity augmentation problem. This inspired us to engineer novel efficient algorithms for the connectivity augmentation problem.
In this work, we implemented recently published approximation algorithms, as well as new heuristic strategies and an exact approach for solving the weighted connectivity augmentation problem in large graphs efficiently. Our greedy heuristic \gwc~excels in solving small link cost instances, while our minimum spanning tree-based algorithm \mstconnect~is the top choice for large cost instances in terms of solution quality, running time, and memory consumption. Additionally, we introduce a novel local search algorithm \ls~that enhances existing solutions -- the first local search algorithm in the literature. Lastly, we engineer an exact solver \ILP, which is also able to compete regarding running time on real-world instances with small link costs.
We conducted experiments comparing our implementations with our best faith implementation of Watanabe et al.'s state-of-the-art solvers. Our algorithms  significantly surpass their results in solution quality, running time, and memory consumption and are very close to optimal solutions. 
Surprisingly, on real-world instances all previous state-of-the-art algorithms, \ie the algorithm by Watanabe \etal, are even outperformed  in terms of running time by our exact solver {\ILP} on the (large number of) instances that it could solve. Instances that the {\ILP} approach could not solve, have also not been solved by the algorithms by Watanabe {\etal}  %
Important future work includes exploring additional local search algorithms and reducing its search space through additional pruning. Additionally, we want to improve the scalability of these algorithms even further.  Given the good results of our algorithm we plan to release it as open source.

\newpage

\bibliographystyle{plain}
\bibliography{quellen}

\newpage
\appendix
\section{Data Structures}\label{sec:appendix_data_structure}

All minimum cuts of a graph $G$ can be represented as a (potentially significantly smaller) cactus graph $C$, it is sufficient to do computations on the cactus graph $C$. To be able to give an augmentation for the original graph $G$ it must be known. To achieve this, $G$ is stored in adjacency list representation along with an array modelling the function $\Pi:V(G)\to V(C)$ and using vertex IDs as indices.
Additionally, the link set $L$ must be transferred to a link set $L_C$ for the cactus graph $C$. A link $l=(u,v)$ is translated to a link $l_C=\Pi(l):=(\Pi(u),\Pi(v))$ in the cactus graph. %
To be able to reverse this function and obtain a link in $G$ again, the endpoints of the original link $l$ are stored as well. 
This can lead to parallel links $g,h\in L,g\neq h$ with $\Pi(g)=\Pi(h)$ in the link graph $G_{L_C}$ (with not necessarily equal weight). We drop all parallel links in $G_{L_C}$ except for one of smallest~weight.
We store the link set $L_C$ in an adjacency matrix and only need to keep an arbitrary link of minimum weight per vertex pair.

\textbf{Dynamic Cactus / Updating Cactus.}
For a dynamic cactus representation we use a similar approach for updating the cactus as proposed by Henzinger, Noe and Schulz in~\cite{dynamic-cactus}, where a union find data structure is used to keep track of the function $\Pi$ that associates each vertex of $G$ with a vertex of the cactus. To be able to provide more information, this data structure is extended and modified. For implementation details, we refer to~\cite{wca-thesis}. By using this data structure we are able to compute the number of cuts that a given link crosses within $\mathcal{O}(|V_c|)$. 
CS: removed this as we seem not to use it 

\section{Cactus Graph Generation}\label{sec:generator}
Graphs where the cactus graph representing all minimum cuts has a complex structure are rare in real-world instances like technical or social networks. A graph generator can be used to generate complex graphs of different sizes and evaluate algorithms for such cases. This section describes an algorithm that is able to generate graphs with given properties, namely the number of vertices and the number of cycles.

\textbf{Generating Cactus Graphs.}
Given two integers $n,c\in\mathbb{N}$, $n>c$, the goal is to generate a cactus graph $C$ with $n$ vertices and $c$ cycles. This is done by generating the cycles iteratively. The average number of vertices per cycle is $n/c$. To get a larger amount of possible graphs the number of vertices per cycle is randomly distributed around the average $n/c$. A Poisson distribution turned out to yield a higher variety of graphs than a uniform distribution. To ensure that the correct number of cycles will be achieved, the distribution range is bounded such that, considering cycles already generated, at least one vertex for every remaining cycle is available. The first generated cycle is used as the base graph. Each consecutive cycle must additionally use an existing vertex to connect to the base graph. This vertex is chosen uniformly among existing vertices.

\textbf{Graph with given Cactus Graph.}
Given a cactus graph $C$, one might ask how a graph $G$ of which $C$ represents all minimum cuts could look like. Trivially $G$ could be equal to $C$. Different graphs could be constructed by reversing the process of edge contractions during the computation of a cactus graph. In particular, each vertex of the cactus graph $C$ could be replaced by a dense subgraph. Let $k$ be the desired connectivity. Then, each vertex can be replaced by an at least $(k+1)$-connected subgraph while each link is replaced by $k$ unweighted links in case of a tree edge or by $k/2$ links in case of a cycle edge between corresponding dense subgraphs. However, for all algorithms considered in this paper, neither the structure of the original graph nor the connectivity $k$ matter as they are abstracted in preprocessing steps. Therefore, only the simplest case of cactus graphs with connectivity $k=2$ is considered.
\vfill
\section{Instance Details}

\begin{table}[H]
	\caption{Properties of all graphs $G$ and their corresponding cactus graph $C$ used in the evaluation.}
	\label{tab:apx:instances}
		\resizebox{\linewidth}{!}{
		\centering
		\hspace{1.9cm}
		\begin{minipage}{1.28\linewidth}
	    \begin{tabular*}{\textwidth}{llrrrr}
		class &graph & $|V(C)|$ & $|E(C)|$ & $|V(G)|$ & $|E(G)|$\\
		\cmidrule(r){1-6}
		\textbf{Real-World}
		&coAuthorsCiteseer&\numprint{30322}&\numprint{30321}&\numprint{227320}&\numprint{814134}\\
		&preferentialAttachment&\numprint{28530}&\numprint{28529}&\numprint{100000}&\numprint{499985}\\
		&delaunay\_n21&\numprint{23719}&\numprint{23718}&\numprint{2097152}&\numprint{6291408}\\
		&luxembourg.osm&\numprint{23077}&\numprint{23076}&\numprint{114599}&\numprint{119666}\\
		&kkt\_power&\numprint{22388}&\numprint{22387}&\numprint{2063494}&\numprint{6482320}\\
		&delaunay\_n20&\numprint{11740}&\numprint{11739}&\numprint{1048576}&\numprint{3145686}\\
		&coPapersDBLP&\numprint{10244}&\numprint{10243}&\numprint{540486}&\numprint{15245729}\\
		&as-22july06&\numprint{7999}&\numprint{7998}&\numprint{22963}&\numprint{48436}\\
		&\textit{hugetric-00000}&\numprint{6617}&\numprint{8040}&\numprint{5824554}&\numprint{8733523}\\
		&coPapersCiteseer&\numprint{6372}&\numprint{6371}&\numprint{434102}&\numprint{16036720}\\
		&delaunay\_n19&\numprint{5977}&\numprint{5976}&\numprint{524288}&\numprint{1572823}\\
		&PGPgiantcompo&\numprint{5513}&\numprint{5512}&\numprint{10680}&\numprint{24316}\\
		&vsp\_vibrobox\_scagr7-2c\_rlfddd&\numprint{3956}&\numprint{3955}&\numprint{77328}&\numprint{435586}\\
		&vsp\_finan512\_scagr7-2c\_rlfddd&\numprint{3936}&\numprint{3935}&\numprint{139752}&\numprint{552020}\\
		&vsp\_sctap1-2b\_and\_seymourl&\numprint{3288}&\numprint{3287}&\numprint{40174}&\numprint{140831}\\
		&finan512&\numprint{3073}&\numprint{3072}&\numprint{74752}&\numprint{261120}\\
		&delaunay\_n18&\numprint{2929}&\numprint{2928}&\numprint{262144}&\numprint{786396}\\
		&vsp\_south31\_slptsk&\numprint{2710}&\numprint{2709}&\numprint{39668}&\numprint{189914}\\
		&vsp\_model1\_crew1\_cr42\_south31&\numprint{2561}&\numprint{2560}&\numprint{45101}&\numprint{189976}\\
		&vsp\_c-30\_data\_data&\numprint{1768}&\numprint{1767}&\numprint{11023}&\numprint{62184}\\
		&power&\numprint{1612}&\numprint{1611}&\numprint{4941}&\numprint{6594}\\
		&delaunay\_n17&\numprint{1484}&\numprint{1483}&\numprint{131072}&\numprint{393176}\\
		&af\_shell9&\numprint{1276}&\numprint{1275}&\numprint{504855}&\numprint{8542010}\\
		&vsp\_bump2\_e18\_aa01\_model1\_crew1&\numprint{1212}&\numprint{1211}&\numprint{56438}&\numprint{300801}\\
		&\textit{t60k}&\numprint{1136}&\numprint{1332}&\numprint{60005}&\numprint{89440}\\
		&vsp\_p0291\_seymourl\_iiasa&\numprint{942}&\numprint{941}&\numprint{10498}&\numprint{53868}\\
		&ldoor&\numprint{904}&\numprint{903}&\numprint{952203}&\numprint{22785136}\\
		&latin\_square\_10&\numprint{901}&\numprint{900}&\numprint{900}&\numprint{307350}\\
		&delaunay\_n16&\numprint{790}&\numprint{789}&\numprint{65536}&\numprint{196575}\\
		&add32&\numprint{681}&\numprint{680}&\numprint{4960}&\numprint{9462}\\
		&vibrobox&\numprint{625}&\numprint{624}&\numprint{12328}&\numprint{165250}\\
		&add20&\numprint{367}&\numprint{366}&\numprint{2395}&\numprint{7462}\\
		&delaunay\_n15&\numprint{359}&\numprint{358}&\numprint{32768}&\numprint{98274}\\
		&NLR&\numprint{336}&\numprint{335}&\numprint{4163763}&\numprint{12487976}\\
		&G3\_circuit&\numprint{317}&\numprint{316}&\numprint{1585478}&\numprint{3037674}\\
		&vsp\_befref\_fxm\_2\_4\_air02&\numprint{215}&\numprint{214}&\numprint{14109}&\numprint{98224}\\
		&delaunay\_n14&\numprint{181}&\numprint{180}&\numprint{16384}&\numprint{49122}\\
		&audikw1&\numprint{163}&\numprint{162}&\numprint{943695}&\numprint{38354076}\\
		&email&\numprint{156}&\numprint{155}&\numprint{1133}&\numprint{5451}\\
		&uk&\numprint{136}&\numprint{135}&\numprint{4824}&\numprint{6837}\\
		&M6&\numprint{132}&\numprint{131}&\numprint{3501776}&\numprint{10501936}\\
		&cage15&\numprint{129}&\numprint{128}&\numprint{5154859}&\numprint{47022346}\\
		&memplus&\numprint{110}&\numprint{109}&\numprint{17758}&\numprint{54196}\\
 		\textbf{Generated Cycles}
 		&cycle-5000 & \numprint{5000} & \numprint{5000}&&\\
 		&cycle-1000 & \numprint{1000} & \numprint{1000} &&\\
 		&cycle-500 & 500 & 500 &&\\
 		&cycle-300 & 300 & 300 &&\\
 		&cycle-200 & 200 & 200 &&\\
 		&cycle-100 & 100 & 100 &&\\
		&cycle-50 & 50 & 50 &&\\
		 \textbf{Generated Stars} 
		&star-5000 & \numprint{5000} & \numprint{4999} &&\\
		&star-1000 & \numprint{1000} & \numprint{999}&&\\
		&star-500 & 500 & 499 &&\\
		&star-300 & 300 & 300 &&\\
		&star-200 & 200 & 199&& \\
		&star-100 & 100 & 99&& \\
		&star-50 & 50 & 49 &&\\
		 \textbf{Generated Cacti} 
		&cactus\{16-20\} & 1000 & 1199&&\\
		&cactus\{11-15\} & 200 & 279&&\\
		&cactus\{06-10\} & 100 & 119&&\\
		&cactus\{01-05\} & 100 & 109&&\\
	\end{tabular*}
\end{minipage}}
\end{table}

\end{document}